\documentclass[final]{siamltex}
\usepackage{framed,algorithm,algorithmic}
\usepackage{amsfonts,amsmath,bm,url}
\long\def\symbolfootnote[#1]#2{\begingroup%
\def\thefootnote{\fnsymbol{footnote}}\footnote[#1]{#2}\endgroup}

\newcommand{\Expect}[1]{\mbox{}{\mathbb{E}}\left[#1\right]}
\newcommand{\ExpectSub}[2]{\mbox{}{\mathbb{E}}_{#1}\left[#2\right]}

\newcommand{\FNorm }[1]{\mbox{}\|#1\|_{\mathrm{F}}  }
\newcommand{\FNormS}[1]{\mbox{}\|#1\|_\mathrm{F}^2}

\newcommand{\TNorm }[1]{\mbox{}\|#1\|_2  }
\newcommand{\TNormS}[1]{\mbox{}\|#1\|_2^2}

\newcommand{\XNorm }[1]{\mbox{}\|#1\|_{\xi}  }
\newcommand{\XNormS}[1]{\mbox{}\|#1\|_{\xi}^2}

\newtheorem{problem}[theorem]{Problem}

\newcommand{\transp}{^{\textsc{T}}}
\newcommand{\trace}{\text{\rm Tr}}
\newcommand{\st}{\text{\rm St}}
\newcommand{\Adj}{\text{\rm Adj}}

\newcommand{\mat}[1]{{\ensuremath{\bm{\mathrm{#1}}}}}

\newcommand{\binset}[2]{C \left( {#1}, {#2} \right)}
\newcommand{\binnum}[2]{\binom{#1}{#2}}

\def\volsamp{\hbox{\rm VolSamp}}

\def\rank{\hbox{\rm rank}}
\def\diag{\hbox{\rm diag}}

\def\b{{\mathbf b}}
\def\e{{\mathbf e}}

\def\u{{\mathbf u}}
\def\v{{\mathbf v}}

\def\matA{\mat{A}}
\def\matB{\mat{B}}

\def\matD{\mat{D}}
\def\matE{\mat{E}}

\def\matI{\mat{I}}
\def\matL{\mat{L}}

\def\matT{\mat{T}}
\def\matU{\mat{U}}
\def\matV{\mat{V}}
\def\matW{\mat{W}}
\def\matX{\mat{X}}
\def\matY{\mat{Y}}
\def\matZ{\mat{Z}}
\def\matSig{\mat{\Sigma}}

\def\matPi{\mat{\Pi}}

\DeclareMathSymbol{\Prob}{\mathbin}{AMSb}{"50}
\newcommand\remove[1]{}

\def\math#1{$#1$}

\def\frac#1#2{{#1\over #2}}

\def\eqan#1{\begin{eqnarray*}
#1
\end{eqnarray*}}

\DeclareMathSymbol{\R}{\mathbin}{AMSb}{"52}

\def\cl#1{{\cal #1}}

\def\x{{\mathbf x}}
\def\y{{\mathbf y}}

\def\b{{\mathbf b}}

\def\norm#1{{\|#1\|}}

\def\ceil#1{{\left\lceil\,#1\,\right\rceil}}

\def\dotfil{\leaders\hbox to 1.5mm{.}\hfill}
\def\RN#1{\setcounter{rmnum}{#1}\uppercase\expandafter{\romannumeral\value{rmnum}}}
\def\rn#1{\setcounter{rmnum}{#1}\expandafter{\romannumeral\value{rmnum}}}

\begin{document}

\title{FASTER SUBSET SELECTION FOR MATRICES AND APPLICATIONS}

\author{
Haim Avron\thanks{Mathematical Sciences Department, IBM T. J. Watson Research Center, Yorktown Heights, NY 10598. Email: \texttt{haimav@us.ibm.com}}
\and
Christos Boutsidis\thanks{Mathematical Sciences Department, IBM T. J. Watson Research Center, Yorktown Heights, NY 10598. Email: \texttt{cboutsi@us.ibm.com}}
}

\maketitle

\begin{abstract}
\noindent We study the following problem of \emph{subset selection} for matrices:
given a matrix $\matX \in \R^{n \times m}$ ($m > n$) and a sampling
parameter $k$ ($n \le k \le m$), select a subset of $k$ columns from
$\matX$ such that the pseudo-inverse of the sampled matrix has as
smallest norm as possible. In this work, we focus on the Frobenius and
the spectral matrix norms.
We describe several novel (deterministic and randomized) approximation algorithms for this
problem with approximation bounds that are optimal up to constant factors.
Additionally, we show that the combinatorial problem of finding a low-stretch spanning tree in an undirected graph
corresponds to subset selection, and discuss various implications of this reduction.
\end{abstract}

\begin{keywords}
Subset Selection, Low-stretch Spanning Trees, Volume Sampling, Low-rank Approximations, $k$-means Clustering, feature selection, Sparse Approximation.
\end{keywords}

\begin{AMS}
15B52, 15A18, 90C27 
\end{AMS}

\pagestyle{myheadings}
\thispagestyle{plain}
\markboth{AVRON AND BOUTSIDIS}{FASTER SUBSET SELECTION FOR MATRICES AND APPLICATIONS}

\section{Introduction}
Given a full rank short-and-fat matrix $\matX \in \R^{n \times m}$ with $m > n$  (typically $m \gg n$)
it is often of interest to \emph{compress} $\matX$ via selecting a subset of its columns. The goal
of such a sampling procedure is to select the columns in a way that the sampled matrix behaves spectrally
similarly to the original matrix, i.e. the singular values of the two matrices are comparable.
Since deleting columns from $\matX$ decreases the singular values monotonically
(this is immediate from the interlacing property of the singular values; see Theorem 8.1.7 on page 396 in~\cite{GV96}), the challenge
is to select the columns that in a sense (which we make precise in the definition below) maximize the spectrum in the sampled matrix.
In particular, we consider the following combinatorial optimization problem
(let $[m] = \{i \in \mathbb{N}: i\le m \}$, i.e. the set of natural numbers $1,2,...m$).

\begin{problem}[Subset Selection for Matrices]\label{def:prob}
Fix $\matX \in \R^{n \times m}$ with $m > n$ and a sampling parameter $k$ with $n \le k \le m$.
Let $\cal S \subseteq$ $ [m] $ denote a set of cardinality at most $k$ for which $\rank(\matX_{\cal S})=\rank(\matX)$, where $\matX_{\cal S} \in \R^{n \times |{\cal S}|}$ contains the
columns of $\matX$ indicated in $\cal S$. Among all such possible choices of $\cal S$, find an ${\cal S}_{opt}$ such that
$\XNorm{ \matX_{ {\cal S}_{opt}}^{\dagger} }$ is minimized, i.e.,
$$ {\cal S}_{opt}  \in \arg \min_{{\cal S} \in {\cal F}(\matX, k) } \XNorm{ \matX_{ \cal S}^{\dagger} }\,,$$
where ${\cal F}(\matX, k) = \{ {\cal S}\,:\,|{\cal S}| \leq k,\,\rank(\matX_{\cal S})=\rank(\matX) \}$. Note that there might be more than one possibility for ${\cal S}_{opt}$ (the minimizer might not be unique). In the above, $\xi = 2, \mathrm{F}$ denotes the spectral or the Frobenius matrix norm, respectively, and $\matX_{ \cal S}^{\dagger}$ denotes the Moore-Penrose pseudo-inverse of $\matX_{ \cal S}$.
\end{problem}

Technically, the above definition corresponds to two different combinatorial optimization problems, one
for $\xi = 2$ and the other for $\xi = \mathrm{F}$.
\remove{The spectral norm version of this problem is known to be NP-hard
(see the $Max$-$MinSingularValue$ part of Theorem 4 in~\cite{CM09}). To the best of our knowledge, there is
currently no known hardness result for the Frobenius norm version. Prior to our work the tightest bounds
known for Problem~\ref{def:prob} were attained by the algorithm described in~\cite{HM07}.
In this work, we present several novel algorithms that achieve the same or tighter
approximations as the method of~\cite{HM07} using \emph{less} operations.
Each algorithm has an advantage over the other
algorithms in certain aspects. Moreover, by proving novel lower bounds for Problem~\ref{def:prob}, we show that the approximation bounds of our best algorithms are optimal (up to constants).
}

Problem~\ref{def:prob} occurs in
numerous situations:
column-based low-rank matrix approximation~\cite{BDM11a,BMD09a};
feature selection in $k$-means clustering~\cite{BM11,BZMD11};
optimal experiment design~\cite{HM07,HM11};
multipoint boundary value problems~\cite{HM07,HM11};
sparse solutions to least-squares regression~\cite{CH92a,Bou11b};
sensor selection in a wireless network~\cite{JB09};
rank-deficient linear least squares~\cite{FK06}, and rank-deficient non-linear least squares~\cite{IKP11}, to name just a few.
We discuss some of these situations in Section~\ref{sec:other}.

However, our initial motivation for investigating Problem~\ref{def:prob} was our observation that
the combinatorial problem of finding a low-stretch spanning tree in an undirected graph~\cite{AKPW95}
corresponds to the Frobenius norm version of Problem~\ref{def:prob}. This connection is new
and might be of independent interest.


We study three aspects of Problem~\ref{def:prob}: algorithms, lower bounds, and applications. We now summarize our contributions in each of these aspects.

\subsection{Our contributions}

\subsubsection{Algorithms}
In Section~\ref{sec:main} we describe five different approximation algorithms for Problem~\ref{def:prob}.
We suggest five different algorithms because no single
algorithm has the lowest operation count; the choice of the most efficient algorithm depends on the actual values of $m$, $n$ and $k$. Our algorithms are considerably faster than the previously known algorithms, and they achieve the same or tighter approximation bounds.
Table~\ref{table:algs} summarizes the algorithms we propose, as well as previously known algorithms for Problem~\ref{def:prob}.

Our first two algorithms, Algorithm~\ref{alg1} and Algorithm~\ref{alg2}, which we describe in Theorem~\ref{thm1} and Corollary~\ref{cor1} (both in Section~\ref{sec:main}), respectively, are especially fast when $k$ is close to $m$, since they form ${\cal S}$ by greedily removing columns. Both algorithms are deterministic. Algorithm~\ref{alg1} in Theorem~\ref{thm1} is designed for the Frobenius  norm case ($\xi = \mathrm{F}$). It requires $O\left( mn^2 + mn(m-k) \right)$ operations, and finds a subset ${\cal S}$ of cardinality $k$ such that
$$
\FNormS{  \matX_{ \cal S}^{\dagger} } \le \frac{m-n+1}{k-n+1} \cdot \FNormS{\matX^{\dagger}}\,.
$$
Notice, for example, that if $k= m - \alpha$, for some small integer $0 < \alpha  \le 0.9 (m-n+1)$, then the approximation bound is $1 + 10\alpha(m-n+1)^{-1}$.

Algorithm~\ref{alg2} in Corollary~\ref{cor1} is designed for the spectral norm case ($\xi = 2$). It's operation count is $O\left( mn^2 + mn(m-k) \right)$ as well. It finds a subset ${\cal S}$ of cardinality $k$ such that
$$
\TNormS{ \matX_{ \cal S}^{\dagger} } \le \left(1+ \frac{n\left(m-k\right)}{k-n+1} \right) \cdot \TNormS{\matX^{\dagger}}\,.
$$
Similarly, if, for example, $k= n+1+ \beta$, for some integer $\beta$ close to $m$ with $0 < \beta  < m-n+1$, then the approximation bound is $1 + n + n (m-n-1) \beta^{-1}$.

The idea of greedily removing columns was previously used by de Hoog and R. Mattheijb in~\cite{HM07}. However, our algorithms are at least a factor of $n$ faster, and in some cases a factor of $n^2$ faster. Furthermore, our algorithms operate on a wider range of matrices: the algorithms in~\cite{HM07} require that all possible column subsets in $\matX$ of size $k$ or larger are non-singular, while our algorithms have no such restriction
(see the paragraph Greedy Algorithms in Section~\ref{sec:prior} for a detailed discussion of the results in~\cite{HM07}).

\begin{table}
\begin{center}
\begin{tabular}{|l|l|l|l|l|l|}
\hline
        & {\bf Sampling}  & {\bf Bound on $\frac{\FNormS{ \matX_{ \cal S}^{\dagger} }}{\FNormS{ \matX^{\dagger} }}$ }  &
 {\bf Bound on $\frac{\TNormS{ \matX_{ \cal S}^{\dagger} }}{\TNormS{ \matX^{\dagger} }}$ } & {\bf Operation count} \\
\hline
\hline
\multicolumn{1}{|l}{{\em Old Algorithms}} & \multicolumn{1}{c}{} & \multicolumn{1}{c}{} & \multicolumn{1}{c}{} & \multicolumn{1}{c|}{}\\
\hline
\hline
Theorem 2 in~\cite{HM07}  & $k \ge n$ & $\frac{m-n+1}{k-n+1}$ & $\frac{m-n+1}{k-n+1} \cdot n$ & $O(mn^3(m-k))$ \\
\hline
Corollary 2 in~\cite{HM07}  & $k \ge n$ & $\frac{m-n+1}{k-n+1} \cdot \frac{n \TNormS{\matX^{\dagger}}}{\FNormS{\matX^{\dagger}}}$ & $1+ \frac{n(m-k)}{k-n+1}$ & $O(mn^3(m-k))$  \\
\hline
Theorem 1 in~\cite{HM11}  & $k \ge n$ & $\frac{m-n+1}{k-n+1} \cdot \frac{n \TNormS{\matX^{\dagger}}}{\FNormS{\matX^{\dagger}}}$ & $1+ \frac{n(m-k)}{k-n+1}$ & $O(mn^3(m-k))$  \\
\hline
Lemma 16 in~\cite{Bou11a}  & $k = n$ & $ \left( 1+ f^2 n \left( m-n \right) \right) $ & $1+ f^2n(m-n)$ & $ O(m n^2 \log_{f} m) $  \\
(Alg. is from~\cite{GE96}) & & & & \\
\hline
Lemma 1 in~\cite{Git11} & $k>$ &&& \\
$\delta=1/2$  & $8 \cdot \tau \cdot n \cdot \log (2n)$ &  No bound & $\frac{2 m }{k}$ & $ O(m n^2 + k) $ \\
\hline
Section A in~\cite{JB09}  & $k \ge n$ &  No bound &  No bound & $ O(m^3) $ \\
\hline
\hline
\multicolumn{1}{|l}{{\em New Algorithms}} & \multicolumn{1}{c}{} & \multicolumn{1}{c}{} & \multicolumn{1}{c}{} & \multicolumn{1}{c|}{}\\
\hline
\hline
Theorem~\ref{thm1} & $k \ge n$ & $\frac{m-n+1}{k-n+1}$ & $\frac{m-n+1}{k-n+1} \cdot n$ & $O(mn^2 + mn(m-k))$  \\
\hline
Corollary~\ref{cor1} & $k \ge n$ & $\frac{m-n+1}{k-n+1} \cdot \frac{n \TNormS{\matX^{\dagger}}}{\FNormS{\matX^{\dagger}}}$ & $1+ \frac{n(m-k)}{k-n+1}$ & $O(mn^2 + mn(m-k))$  \\
\hline
Theorem~\ref{thm2} & $k > n$ & $\frac{( 1 + \sqrt{\frac{m}{k}})^2}{( 1 - \sqrt{\frac{n}{k})^2}}$ & $\frac{( 1 + \sqrt{\frac{m}{k}})^2}{( 1 - \sqrt{\frac{n}{k})^2}}$ & $O(m n^2 k)$  \\
\hline
Theorem~\ref{thm3}& $k \geq $ &&&\\ $\delta = 1/2$ & $32 \cdot n \cdot \ln (4 n)$ &
$4 m  $
& $4 m$
& $O(mn^2 + k \log k) $  \\
\hline
Theorem~\ref{thm4}  &  &  & &\\
$\delta = 1/2$ & $k = n$ & $(1+\eta)(m-n+1)$ & $(1+\eta) n(m-n+1)$ & $O(mn^3/\log(1+\eta)) $  \\
\hline
\end{tabular}
\caption{ Summary of various algorithms for Problem~\ref{def:prob} (our algorithms, as well as previous algorithms).
$\matX \in \R^{n \times m}$ is a full rank matrix. $k$ denotes the number of sampled columns.
$\cal S$ $\subseteq [m]$ has cardinality at most $k$. $\delta$ denotes a failure probability,
which is assumed to be zero if it is omitted from the description.
In the fourth line of the table, $f > 1$ is a parameter which trades accuracy with number of operations.
In the fifth line of the table,
$\tau$ denotes the so-called coherence of $\matX$:
$ \tau =  \frac{m}{n} \max_{i \in [m]} (\matV \matV \transp)_{ii}$, where $\matV\in \R^{m \times n}$ contains the right singular vectors of $\matX$ corresponding
to the top $n$ singular values of $\matX$.
Lemma 1 in~\cite{Git11} assumes that $\matX$ has orthonormal rows; but, this can be extended to arbitrary $\matX$ just by applying
the result to $\matV\transp$.
In the sixth line of the table, the formulation in~\cite{JB09} assumes that $\matX$ is orthonormal, but this can
be generalized as well.
In the last line of the table, $\eta > 0$ is a parameter which trades the accuracy in the bound with the number of operations of the algorithm.
}
\label{table:algs}
\end{center}
\end{table}
\begin{table}
\begin{center}
\begin{tabular}{|l|l|l|}
\hline
&  {\bf $\FNormS{\matX_{\cal S}^{\dagger}} \ge \gamma \FNormS{\matX^{\dagger}};$ $\gamma=$ }
&  {\bf $\TNormS{\matX_{\cal S}^{\dagger}} \ge \gamma \TNormS{\matX^{\dagger}};$ $\gamma=$ } \\
\hline
\hline
\textbf{$k=n$}
& $m/n $ &  $m$      \\
\hline
\textbf{$k>n$, $k = O(n)$}
& $m/k-C$ & $m/k-1$  \\
\hline
\textbf{$k>n$, $k = \omega(n)$}
& $m/k - k/n$ & $m/k-1$        \\
\hline
\end{tabular}
\caption{
Summary of lower bounds for Problem~\ref{def:prob}. By lower bounds, we mean that there
is a matrix $\matX \in \R^{n \times m}$ such that for every ${\cal S}$,
$\XNormS{\matX_{\cal S}^{\dagger}} \ge \gamma \XNormS{\matX^{\dagger}}$, for a value of $\gamma$ shown in the table.
${\cal S} \subseteq [m]$ has cardinality at most $n \le k \le m$.
For $\xi=2$ and $n=k$, the bound is from Lemma 2.2 in~\cite{GTZ97}; the bound for $\xi=\mathrm{F}$ is
an immediate corollary. We prove the other bounds in Section~\ref{sec:opt}. $C$ denotes a constant.
}
\label{table:lower}
\end{center}
\end{table}

Our third algorithm, which we describe as Algorithm~\ref{alg3} in Theorem~\ref{thm3} (Section~\ref{sec:main}), is designed
for cases that $k$ is small, e.g. $k=O(n)$, a case which is common in applications (see Section~\ref{sec:other}).
The algorithm's operation count is $O\left( mn^2 + kn^2m \right)$, and it constructs a subset
$\cal S$ with cardinality at most $k > n$, such that, for both $\xi=2,\mathrm{F}$:
$$ \XNormS{ \matX_{\cal S}^{\dagger} } \le \left( 1 + \sqrt{\frac{m}{k}}\right)^2 \left( 1 - \sqrt{\frac{n}{k}} \right)^{-2} \XNormS{ \matX^{\dagger} }\,.$$
This algorithm is inspired by recent results on approximate decompositions of the identity~\cite{BSS09,BDM11a}.
%
Notice that, for example, if $k = \Theta(n)$, the approximation bound is $1+O(m/k)$.

Our fourth algorithm, which we describe as Algorithm~\ref{alg4} in Theorem~\ref{thm3} (Section~\ref{sec:main}),
is designed for cases where both $m$ and $k$ are large (specifically, $k=\Omega(n \log n)$).
It is especially fast since it is based on randomly sampling columns of the matrix. However,
we do not use uniform sampling, so our bounds are independent of numerical properties
of the matrix. The operation count of the algorithm is $O(mn^2 + k \log k)$.
For a fixed probability parameter $\delta$ ($0 < \delta < 1$), and $k \geq \ceil{32 n \ln (2n/\delta)}$,
the algorithm constructs a subset $\cal S$ of cardinality at most $k$, such that, for both $\xi=2,\mathrm{F}$, and with probability  $1-\delta$,
$$ \XNormS{ \matX_{\cal S}^{\dagger} } \le 4 \cdot m \cdot \XNormS{ \matX^{\dagger} }\,.$$
If $\matX$ has orthonormal rows, then, the operation count drops to $O(mn + k \log k)$, i.e. linear in the size of the input. The analysis of the algorithm is based on the matrix concentration bound of~\cite{RV07}.


Our last algorithm, Algorithm~\ref{alg5} in Theorem~\ref{thm4} (Section~\ref{sec:volume}) is designed for $k=n$.
It is based on the following theoretical contribution (Lemma~\ref{lem:volrand} in Section~\ref{sec:volume}):
if we randomly sample a subset $\cal S$ of cardinality $k \ge n$ with probability proportional to $\det(\matX_{\cal S} \matX\transp_{\cal S})$, then,
\eqan{
\Expect{\FNormS{\matX_{\cal S}^{\dagger}}}\leq\frac{m-n+1}{k-n+1} \cdot \FNormS{\matX^{\dagger}}
\qquad\mbox{and} \hspace{0.3in}
\Expect{\TNormS{\matX_{\cal S}^{\dagger}}} \leq \left(1+\frac{n(m-k)}{k-n+1} \right) \cdot \TNormS{\matX^{\dagger}}.
}
Algorithm~\ref{alg5}  finds a subset
{\cal S} of cardinality $k=n$
such that
$$ \TNormS{ \matX_{\cal S}^{-1} }  \le \FNormS{ \matX_{\cal S}^{-1} }
\le (1+\eta) \cdot \left( m-n+1 \right) \cdot  \FNormS{ \matX^{\dagger} }
\le (1+\eta) \cdot \left( m-n+1 \right) \cdot  n \cdot \TNormS{ \matX^{\dagger}}\,,$$
for any $\eta > 0$ chosen by the user.
This bound is deterministic
but the bound on the number of operations is probabilistic. Specifically, for any $0 < \delta < 1$, we show that with probability $1 - \delta$,
the operation count is $O\left(m  n^3  \log \delta^{-1} \log^{-1}{(1+\eta)} \right)$.

Our volume-sampling-based algorithm for the subset
selection problem can be viewed as a complementary result to the volume-sampling-based algorithms designed before for low-rank matrix approximation~\cite{DRVW06}.
In low-rank matrix approximation, the subspace spanned by the columns that are selected by volume sampling contains a rank $k$ matrix that approximates the best rank $k$
matrix computed via the SVD; in our case, the objective is different but we show that volume sampling gives useful results as well.

\subsubsection{Lower Bounds}
By lower bounds, we mean that there exists a matrix $\matX \in \R^{n \times m}$ such that for every ${\cal S}$ of cardinality $k \ge n$,
for $\xi = 2$ or $\xi = \mathrm{F}$,
we have $\XNormS{\matX_{\cal S}^{\dagger}} \ge \gamma \XNormS{\matX^{\dagger}}$ for some value of $\gamma$ which we call \emph{lower bound}.
We develop such lower bounds via, first, relating the subset selection problem to the so-called column-based matrix reconstruction problem~\cite{BDM11a}, and then, employing existing lower bounds~\cite{BDM11a} for column-based matrix reconstruction.
We present  these results in Section~\ref{sec:opt}; a summary of lower bounds appears in Table~\ref{table:lower}.
Our lower bounds indicate that some upper bounds of de Hoog and Mattheij~\cite{HM07,HM11} as well as ours are the best possible up to constant factors. This resolves an open question in~\cite{HM07,HM11}.

An alternative way to study the optimality of our algorithms is to develop lower bounds of the form
$\XNormS{\matX_{\cal S}^{\dagger}} \ge \gamma \XNormS{\matX_{{\cal S}_{opt}}^{\dagger}}$.
However, we were unable to prove such bounds, so we leave this as an interesting open question for future investigation.

\subsubsection{Applications}
In Section~\ref{sec:ap}, we study the connection between low-stretch spanning trees and subset selection.
Using a result by Spielman and Woo~\cite{SW09}, we prove that the stretch of any tree in an undirected graph equals the Frobenius norm squared of the pseudo-inverse of the sampled matrix that arises by sampling columns from an orthonormal matrix which is a basis for the row space of the so-called
node-by-edge incidence matrix of the graph. This incidence matrix contains as many columns as edges in the graph; so,
sampling columns from this matrix corresponds to sampling edges from the graph.
We then use this reduction to develop novel algorithms for constructing spanning trees with low stretch in undirected graphs.
Unfortunately, our algorithms are worse than the available state-of-the-art~\cite{EEST05,ABN08,KMP11}.
We believe, however, that the connection is interesting and might be useful to shed new light on the combinatorial problem of finding a low stretch spanning tree in an undirected graph.

In Section~\ref{sec:other} we use the subset selection algorithms of this paper to design novel algorithms for three other problems involving sub-sampling:
column-based low-rank matrix reconstruction, sparse solution of least-squares problems, and feature selection in $k$-means clustering.

\subsection{Related Work}\label{sec:prior}
We now provide a comprehensive summary of known results regarding Problem~\ref{def:prob} and we comment
on two related subset selection problems studied in the literature.

\subsubsection{Greedy Algorithms}
In~\cite{HM07} de Hoog and Mattheij propose the following algorithm for the Frobenius norm version of Problem~\ref{def:prob}.
The idea is to proceed by removing one column from $\matX$ at a time.
In the first iteration of the algorithm, they remove the column with index $i_1$, where
$$ i_1 = \arg \min_{i=1,...,m} \trace\left(  \left( \matX\matX\transp - \x_i \x_i\transp \right)^{-1}  \right). $$
Let $\matX_1 \in \R^{n \times (m-1)}$ be the matrix obtained after removing the $i_1$th column of $\matX$.
In the second iteration of the algorithm, they remove the column with index $i_2$ such that,
$$ i_2 = \arg \min_{i=1,...,m-1} \trace\left(  \left( \matX_1 \matX_1\transp - \x_i \x_i\transp \right)^{-1}  \right), $$
and so on, until $m-k$ columns are removed.

A straightforward implementation of this idea requires
$O(m n^3 (m-k))$ operations.
However, one can use the
Sherman-Morrison formula for rank one updates to the inverse of a matrix and improve the operation count to $O(n^3 + m n^2 (m-k))$.

Notice that the algorithm just described assumes (implicitly)   that in all the iterations, removing a single column
does not result in a rank deficient matrix; otherwise, for an iterate $\matX_j$ ($j=1,...,m-k$) and a column $\x_i$ ($i=1,...,m$) whose removal
will result in a rank deficient matrix, the inverse of $\matX_j\matX_j\transp - \x_i \x_i\transp$ is not defined.
In~\cite{HM07} it is shown that this algorithm achieves the bound
$\FNormS{\matX^{\dagger}_{\cal S}}\leq \frac{m - n + 1}{k - n + 1} \cdot \FNormS{\matX^{\dagger}}.$
However, the assumption just mentioned is not true in general.

Our algorithms of Theorem~\ref{thm1} and Corollary~\ref{cor1} use the greedy removal idea as well. However, they find the columns to
be removed in a different way. Our algorithms are substantially faster (at least a factor of $n$, and a factor of $n^2$ in some cases) than the algorithm of~\cite{HM07}. Additionally, our algorithms efficiently detect columns whose removal results in a rank deficient matrix, and avoid removing them. So, our algorithms
work for any full-rank matrix $\matX$, without any restriction.

We also mention that Theorem 1 in~\cite{HM11} describes a similar
greedy deterministic algorithm with comparable operation count but slightly worse approximation bounds than the algorithm of~\cite{HM07}
(see Table~\ref{table:algs} for the precise statement of these results). On the positive side, this algorithm works for any $\matX$.

\subsubsection{Rank Revealing Factorizations}
The subset selection problem that we study in this paper
has deep connections, which we do not explain in detail, with the so-called Rank-Revealing
QR~\cite{GE96} (and also see~\cite{BMD09a,CH94} for a summary of available RRQR algorithms)
and Rank-Revealing LU~\cite{GM04}
factorizations.

Worth special mention is the seminal work
of Gu and Eisenstant~\cite{GE96} on Strong Rank-Revealing QR (RRQR). Algorithm 4 and Theorem~3.2 of~\cite{GE96} are stated for matrices with at least as many rows as column, but they can be easily adapted to the case where there are at least as many columns than rows (see Lemma 15 of~\cite{Bou11a} or Equation 3.1 of \cite{Ipsen}). That is, these algorithms provide a numerically stable way to
compute a subset ${\cal S}$ of cardinality  $k \le n$ with bounds on all the non-zero singular values of $\matX_{\cal S}$. Specifically, when $k=n$ then
for $i=1,...,n$ this approach provides the following bound,
$$ \sigma^2_i (\matX_{\cal S}) \ge \frac{ \sigma^2_i(\matX) }{ 1+f^2 n (m-n) }\,.$$
By applying the inequality to $i=n$ we have the following bound,
$$ \TNormS{ \matX_{ \cal S}^{-1} } \le \left(1+ f^2n\left(m-n\right) \right) \cdot \TNormS{\matX^{\dagger}}\,.$$
By summing up the bounds on each singular value we get the following bound,
$$ \FNormS{ \matX_{ \cal S}^{-1} } \le \left(1+ f^2n\left(m-n\right) \right) \cdot \FNormS{\matX^{\dagger}}\,.$$
For $f > 1$ and $k=n$, the operation count of this method is $O(mn^2 \log_fm)$.

Rank revealing approaches can only be used to sample $k \le n$ columns;
extending these approaches to sample arbitrary $k \ge n$ columns, which is the focus of this paper, is not obvious.


\subsubsection{Incoherent Subset Selection}
A recent result by Gittens~\cite{Git11} studies the subset selection problem in the context of the so-called coherence of
$\matX$. The algorithm uses random sampling of columns. Gittens shows that this simple algorithm gives competitive bounds for matrices which have low coherence.

\subsubsection{Approximation via Convex Relaxation}
Joshi and Boyd~\cite{JB09} explored the use of convex relaxation to solve Problem~\ref{def:prob}: initially,
they maximize the norm (spectral or Frobenius) of $(\matX \matE)^{\dagger}$ where $\matE \in \R^{m \times m}$ is a diagonal matrix with diagonal entries that are inside the interval $[0,1]$. This is a convex program, which can be solved, for example, via an interior point algorithm. Note that Problem~\ref{def:prob} corresponds to maximizing the norm of $(\matX \matD)^{\dagger}$ where $\matD \in \R^{m \times m}$ is a diagonal matrix with diagonal entries that are either $0$ or $1$, and setting ${\cal S} = \{i\,:\,\matD_{ii} = 1 \}$. So, to get a feasible solution for Problem~\ref{def:prob}, Joshi and Boyd suggest a rounding scheme to get strictly $0$ or $1$ weights. No theoretical results are reported but the method is shown to perform well in practice.

\subsubsection{Maximum-volume Subsets} Theorem 1 in de Hoog and Mattheij~\cite{HM07} shows that, for $k \ge n$, if $\matX_{\cal S}$
maximizes $\det\left(\matX_{\cal T}\transp \matX_{\cal T}  \right)$ among all possible subsets $\cal T$ of cardinality $k$, then
the following two bounds hold,
$ \TNormS{\matX_{\cal S}^{\dagger}}\le \left( 1 + n\left(m-k\right)/(k-n+1) \right) \cdot \TNormS{\matX^{\dagger}};$ and,
$ \FNormS{\matX_{\cal S}^{\dagger}}\le \left(m-n+1\right)/\left(k-n+1\right) \cdot n \cdot \TNormS{\matX^{\dagger}}.$
Similar spectral norm bounds for $k=n$ were shown before in Eqn.~(2.4) of Theorem 2.2 of~\cite{HP92},
Lemma 3.4 ($\mu=1$, where $\mu$ is a parameter in the lemma) in~\cite{Pan00}, Lemma 2.1 of~\cite{GTZ97}, and Algorithm 3 in~\cite{GE96}.
A similar Frobenius norm bound for $k=n$ was shown before in Eqn.~(2.13) of Theorem 2.3 of~\cite{HP92}.
Notice that all these results do not imply any algorithm other than the naive procedure
of testing all the  $\binnum{m}{k}$ possible subsets of cardinality $k$ (this procedure has an exponential operation count).

Our Lemma~\ref{lem:volrand} proves a similar result, which
states that if one samples $\cal S$ with probability proportional to $\det\left(\matX_{\cal S}\transp \matX_{\cal S}  \right)$,
then, the same bounds hold in expectancy. Now, recent polynomial-time implementations ($O(m n^3)$ operations) of such determinant-based random sampling~\cite{DR10,GK12} allow us to design efficient algorithms ($O(m n^3)$ operations with high probability) that achieve only slightly larger bounds.

Finally, note that the strong RRQR algorithm of~\cite{GE96} finds a local maximum-volume subset. By local maximum-volume subset, we mean that the volume of the subset found is always bigger than the volume of any subset obtained by interchanging a single column.

\subsubsection{Computational Complexity of Subset Selection}
In~\cite{CM09}, Civril and Magdon-Ismail study the spectral norm version of Problem~\ref{def:prob}, as well as three other similar subset selection problems, from a complexity theory point of view. They show that these
problems are NP-hard.  They give special emphasis to the problem of finding a subset $\cal S$ for which $\matX_{\cal S}$ has maximum volume, i.e. $\det(\matX_{\cal S} \matX^\top_{\cal S})$ is maximized. As we discussed above, the problem of finding the subset with maximum volume is connected to Problem~\ref{def:prob}.

The computational complexity of finding a maximum volume subset was also investigated in the computational geometry literature. The problem is stated differently: finding a large simplex in a V-polytope. NP-hardness was established in~\cite{Packer02}, and exponential inapproximability was established in~\cite{Koutis06}.



\subsubsection{Variants of the Subset Selection Problem}
Other variants of subset selection have been studied extensively in numerical linear algebra and computer science. Most of this work focused on spectral norm and
the case of $\matX_{\cal S}$ containing \emph{rescaled} columns from $\matX$~(Theorem 3.1 in~\cite{RV07}; Theorem 11 in~\cite{Zou11}; Theorem 3.1 in~\cite{BSS09})
or $\matX_{\cal S}$ containing \emph{linear combinations} of columns from $\matX$~(Lemma 3.15 of~\cite{MRT11}; Lemma 6 of~\cite{Sar06}; Theorem 1.3 of~\cite{Tro11}). We should note that all these results give much better approximation bounds than our bounds for the spectral norm version of Problem~\ref{def:prob}. For example, the deterministic algorithm in Theorem 3.1 in~\cite{BSS09}, for any $\epsilon > 0$, selects and appropriately rescales $O(n / \epsilon^2)$ columns from $\matX$ and guarantees an approximation bound $1+\epsilon$.


Finally, we mention that all these algorithms have found many applications in numerous problems involving subsampling:
least-squares regression~\cite{AMT10};
column-based low-rank matrix approximation~\cite{BMD09a};
spectral graph sparsification~\cite{SS08}; and, dimensionality reduction in clustering~\cite{BZMD11}.

\subsubsection{Restricted Invertibility}
Bourgain and Tzafriri restricted invertibility result~\cite{BT87} states that there exists a universal constant $C$ such that for every square invertible matrix $\matA \in \R^{n \times n}$ whose columns have unit $\ell_2$ norm, one can find a subset ${\cal S} \subseteq [n]$ of cardinality at least $Cn/\TNormS{\matA}$ such that
$ \TNorm{\matA_{\cal S}} \cdot \TNorm{\matA_{\cal S}^{\dagger}} \leq \sqrt{3}\,.$
Given $\matA$, finding such a subset ${\cal S}$ is another variant of column subset selection.
Tropp gave the first polynomial (randomized) algorithm for restricted invertibility~\cite{Tropp09}.
A deterministic algorithm was recently suggested by Spielman and Srivastava~\cite{SS12}.
However, restricted invertibility deals with selecting fewer than $n$ columns that maximize the smallest non-zero
singular value, while Problem~\ref{def:prob} deals with selecting at least $n$ columns so that the matrix is full-rank, and the smallest singular value is as large as possible. So, the problems are similar, but different.

\section{Preliminaries}\label{sec:pre}


\subsection{Basic Notation}
We use $[n]$ to denote the set $\{1,\dots,n\}$.
We use \math{\matX,\matY\ldots} to denote matrices;
\math{\x,\y\ldots} to denote column vectors.
We denote the columns of $\matX \in \R^{n \times m}$ by $\x_1,\x_2,\ldots,\x_m \in \R^n$; $\x_i$ is column $i$ of $\matX$.
$\matI_{m}$ is the $m \times m$
identity matrix;  $\bm{0}_{n \times m}$ is the $n \times m$ matrix of zeros;
$\bm{e}_i$ is the $i$th standard basis vector (whose dimensionality will be clear from the context): all entries are zero except the $i$th entry which equals one.
$\matX_{ij}$ or $\left( \matX \right)_{ij}$ denotes the $(i,j)$th element of $\matX$. $\v_{ij}$ denotes the $j$th element of a vector $\v_i$.
Logarithms are base two. We abbreviate ``independent identically distributed'' to ``i.i.d''.
Finally, for a set $A$, we denote by $\binset{A}{k}$ the set of all subsets of $A$ of cardinality $k$.

\subsection{Sampling Columns}
In the context of Problem~\ref{def:prob}, $\cal S$ is a set of cardinality $1 < k \le m$, which contains some subset
of the natural numbers from $1,2,...,m$ (repetition of numbers is \emph{not} allowed).
$\matX_{\cal S}$ contains the columns of some matrix
$\matX \in \R^{n \times m}$, indicated in $\cal S$; sometimes we will use $(\matX)_{\cal S}$ to denote the same matrix.
The columns of $\matX_{\cal S}$ are ordered consistently with their order in $\matX$:
if $i,j$ are elements from $\cal S$ and $i<j$, then, the $i$th column of $\matX$ will appear before the $j$th column of $\matX$ in $\matX_{\cal S}$. Finally, $\matX\transp_{\cal S}$ means $(\matX_{\cal S}) \transp$ and $\matX^{\dagger}_{\cal S}$ means $(\matX_{\cal S})^{\dagger}$.

\subsection{Singular Value Decomposition} \label{chap24}
The (thin) Singular Value Decomposition (SVD) of  $\matX \in \R^{n \times m}$ of rank $\rho = \rank(\matX)$ is:
\begin{eqnarray*}
\label{svdA} \matX
         = \underbrace{\left(\begin{array}{cc}
             \matU_{r} & \matU_{\rho-r}
          \end{array}
    \right)}_{\matU \in \R^{n \times \rho}}
    \underbrace{\left(\begin{array}{cc}
             \matSig_{r} & \bf{0}\\
             \bf{0} & \matSig_{\rho - r}
          \end{array}
    \right)}_{\matSig \in \R^{\rho \times \rho}}
    \underbrace{\left(\begin{array}{c}
             \matV_{r}\transp\\
             \matV_{\rho-r}\transp
          \end{array}
    \right)}_{\matV\transp \in \R^{\rho \times m}},
\end{eqnarray*}
with singular values \math{\sigma_1\ge\ldots\sigma_r\geq\sigma_{r+1}\ge\ldots\ge\sigma_\rho > 0}.
Here, $r$ is some rank parameter $1 \le r \le \rho$.
We will often denote $\sigma_1$ as $\sigma_{\max}$ and $\sigma_{\rho}$ as $\sigma_{\min}$, and will use $\sigma_i\left(\matX\right)$ to denote the $i$-th singular value of $\matX$ when the matrix is not clear from the context. The matrices
$\matU_r \in \R^{n \times r}$ and $\matU_{\rho-r} \in \R^{n \times (\rho-r)}$ contain the left singular vectors of~$\matX$; and, similarly, the matrices $\matV_r \in \R^{m \times r}$ and $\matV_{\rho-r} \in \R^{m \times (\rho-r)}$ contain the right singular vectors of~$\matX$.
Finally, we repeatedly use the following column representation for the matrix $\matV$:
$\matV\transp = \matY = [\y_1, \y_2,...,\y_m]$. Here, the $\y_i$'s are vectors in $\R^{\rho}$.

\subsection{Moore-Penrose Pseudo-inverse} Let $\matX \in \R^{n \times m}$ with SVD $\matX = \matU \matSig \matV\transp$.
Then, $\matX^{\dagger} = \matV \matSig^{-1} \matU\transp \in \R^{m \times n}$
denotes the Moore-Penrose pseudo-inverse of $\matX$ ($\matSig^{-1}$ is the inverse of $\matSig$).
\begin{lemma}[Fact 6.4.12 in~\cite{Bernstein05}]
\label{lem:pseudo}
Let $\matA \in \R^{m \times n}, \matB \in \R^{n \times \ell}$, and assume that $\rank(\matA) = \rank(\matB) = n$. Then, \math{(\matA\matB)^{\dagger}=\matB^{\dagger}\matA^{\dagger}}.
\end{lemma}

\begin{lemma}\label{lem:thm2aux}
Let $\matA \in \R^{n \times m}$ be a full rank matrix with $m\geq n$. Let $\matB$ be an invertible $m \times m$ matrix. Then
$$ \TNorm{(\matA \matB)^{\dagger}} \leq \TNorm{\matA^{\dagger}} \cdot \TNorm{\matB^{-1}}\,.$$
\end{lemma}
\begin{proof}
\eqan{
\TNorm{(\matA \matB)^{\dagger}} = (\sigma_{\min}(\matA \matB))^{-1} = (\sigma_{\min}(\matB\transp \matA\transp))^{-1}
& = & \left(\min_{\x \neq 0}\frac{\TNorm{\matB\transp \matA\transp \x}}{\TNorm{\x}}\right)^{-1}\\
& \leq & \left(\min_{\x \neq 0}\frac{\TNorm{\matB\transp \matA\transp \x}}{\TNorm{\matA\transp \x}}\cdot\frac{\TNorm{\matA\transp \x}}{\TNorm{\x}}\right)^{-1} \\
& \buildrel{(*)}\over{\leq} & \left(\min_{\x \neq 0}\frac{\TNorm{\matB\transp \matA\transp \x}}{\TNorm{\matA\transp \x}}\cdot\min_{\x \neq 0}\frac{\TNorm{\matA\transp \x}}{\TNorm{\x}}\right)^{-1} \\
& \leq & \left(\min_{\y \neq 0}\frac{\TNorm{\matB\transp \y}}{\TNorm{\y}}\cdot\min_{\x \neq 0}\frac{\TNorm{\matA\transp \x}}{\TNorm{\x}}\right)^{-1} \\
& = & (\sigma_{\min}(\matB\transp) \cdot \sigma_{\min}(\matA\transp))^{-1} \\
& = & \TNorm{\matA^{\dagger}} \cdot \TNorm{\matB^{-1}}
}
In \math{(*)} we use the fact that $\matA \transp$ is a full rank matrix with more rows than columns (so $\TNorm{\matA\transp \x}\neq 0 $ for $\x\neq 0$).
\end{proof}

\subsection{Column Exchanges and Cramer's rule}
For a matrix $\matA$, an index $i$ and a vector $\v$, we denote by
$\matA(i \rightarrow \v)$ the matrix obtained after replacing the
$i$th column of $\matA$ by $\v$. Notice that for square matrices $\matA$ and $\matB$ of
the same dimension we have $\det(\matA) \det(\matB(i \rightarrow \v)) = \det( (\matA \matB)(i \rightarrow \matA \v) )$.
For an invertible square matrix $\matA$, recall Cramer's rule, which gives a formula for computing the components of
$\x=\matA^{-1}\b$ in terms of determinants. In our notation, the rule states
that $\x_{i}$, the $i$th position in $\x$, is
$$\x_{i}=\frac{\det(\matA(i \rightarrow \b))}{\det(\matA)}\,.$$


\subsection{Volume Sampling} Let $\matX$ be a full rank matrix of dimensions $n \times m$ with $m \geq n$, and let $n \leq k \leq m$ be some integer. Given a subset ${\cal S} \in \binset{[m]}{k}$ define the probability of ${\cal S}$ by
$$ P_{\cal S}= \frac{ \det\left( \matX_{\cal S} \matX_{\cal S}\transp \right) }{ \sum_{{\cal T}\in\binset{[m]}{k}} \det\left( \matX_{\cal T} \matX_{\cal T}\transp \right) }\,. $$
The values $\{P_{\cal S}\}_{{\cal S}\in\binset{[m]}{k}}$ define a distribution over the sets
in $\binset{[m]}{k}$. We denote this distribution by $\volsamp(\matX, k)$. That is, we write ${\cal S} \sim \volsamp(\matX, k)$ to denote that ${\cal S}$ is a random subset which assumes value in $\binset{[m]}{k}$, whose distribution is defined by $$\Pr({\cal S} = {\cal T}) = P_{\cal T}\,.$$
We call this sampling distribution \emph{volume sampling} due to the fact that $\det(\matX_{\cal S} \matX_{\cal S}\transp)^{1/2}$ is the volume of the parallelpiped defined by the rows of $\matX_{\cal S}$. Notice that if $\matA$ is a square non-singular matrix then $\volsamp(\matA \matX, k) = \volsamp(\matX, k)$ (this follows from the fact that $\det((\matA \matX)_{\cal S}(\matA \matX)\transp_{\cal S}) = \det(\matA)^2 \det(\matX_{\cal S}\matX\transp_{\cal S})$, for every $\cal S$).

An efficient algorithm for sampling a set from $\volsamp(\matX, n)$ was first suggested by Deshpande and Rademacher~\cite{DR10}. This algorithm was recently improved by Guruswami and Sinop, who showed how
to sample such a subset with $O(n^{3}m)$ operations~\cite{GK12}.
There is currently no algorithm for sampling from $\volsamp(\matX, k)$ for an arbitrary $k \geq n$.

\subsection{Other Known Results} In addition we use the following two known results.
\begin{lemma}[Special case of the Cauchy-Binet formula] Let $\matA \in \R^{n \times m}, \matB \in \R^{n \times m}$, and $m\ge n$. Then,
$$ \det(\matA \matB\transp) = \sum_{{\cal S} \in \binset{[m]}{n}} \det(\matA_{\cal S}) \det( \matB\transp_{\cal S})\,.$$
\end{lemma}
\begin{lemma}[Theorem 1.2.12 in~\cite{HJ85}]
\label{GK12-lem6}
Let $\lambda_1 \ge \lambda_2 \ge ... \ge \lambda_m,$ denote the eigenvalues of $\matA \in \R^{m \times m}$.
Let $1 \le k \le m$.
Then,
$$ \sum_{{\cal S} \in \binset{[m]}{k}} \det\left( \matA_{{\cal S},{\cal S}} \right) = \sum_{{\cal S} \in \binset{[m]}{k}} \prod_{i \in {\cal S}} \lambda_i.$$
Here, $\matA_{{\cal S},{\cal S}} \in \R^{k \times k}$ denotes the submatrix of $\matA$ corresponding to the rows and the columns in ${\cal S}$
$\subseteq [m]$, which has cardinality $k$.
\end{lemma}

Lemma~\ref{GK12-lem6} expresses the $k$-th elementary symmetric function of the eigenvalues of $\matA$ as
the sum of the $k$-by-$k$ principal minors of $\matA$.

\section{Algorithms}\label{sec:main}

\subsection{Deterministic Greedy Removal (Frobenius norm)}\label{sec:greedy_removal}

This section describes an algorithm based on the same greedy removal strategy as in~\cite{HM07},
but it is faster, since it exploits the SVD decomposition of the matrix and the ability to quickly update it.
Additionally,
our algorithm efficiently detects columns whose removal results in a rank deficient matrix, and avoids removing them
(see the discussion in Section~\ref{sec:prior}).
We prove that our algorithm achieves the same approximation bounds as in~\cite{HM07}.
The proof of~\cite{HM07} does not apply to our algorithm, since~\cite{HM07} assumes implicitly that in all the iterations,
removing a single column does not result in a rank deficient matrix.

A complete pseudo-code description of our algorithm appears as Algorithm~\ref{alg1}. We now explain the steps of the algorithm.
To facilitate the description of the algorithm, we assume the columns in $\matX_{\cal S}$ are indexed using their index in $\matX$. Our algorithm constructs ${\cal S}$ by iteratively removing columns. That is, we start with a complete subset ${\cal S}_0 = [m]$. Then we proceed with $m-k$ iterations, since the goal is to select $k$ columns. We reserve the index
$i$ to refer to the iterations of the algorithm; so, $i=1,2,\dots,m-k$.

Each iteration $i$ of the algorithm starts with some ${\cal S}_{i-1} \subseteq [m]$ of cardinality $m-i+1$, and removes  one index from it, to obtain ${\cal S}_{i} \subset {\cal S}_{i-1}$ of cardinality $m-i$. Our algorithm does exactly $m-k$ iterations, hence it returns ${\cal S}_{m-k}$ of cardinality $k$.
Additionally, in each iteration we maintain an SVD of the current subset,
$\matX_{{\cal S}_i} = \matU^{(i)} \matSig^{(i)} \matY^{(i)}$. We denote the singular values of $\matX_{{\cal S}_i}$ by $\sigma^{(i)}_1,\ldots,\sigma^{(i)}_n$, and the columns of $\matY^{(i)}$ by $\{\y^{(i)}_r\}_{r \in {{\cal S}_i}}$.

Our algorithm begins by computing the SVD of $\matX_{{\cal S}_0} = \matX$. Then, for $i=1,2,\dots,m-k$,
iteration $i$ has two stages:
\begin{enumerate}
\item Finding an index $j_i$ to remove. We then set ${\cal S}_{i} = {\cal S}_{i-1} - \{j_i\}$.
\item Updating the SVD $\matX_{{\cal S}_{i}} = \matU^{(i)} \matSig^{(i)} \matY^{(i)}$. The algorithm needs only $\matSig^{(i)}$ and $\matY^{(i)}$ (no need to downdate $\matU^{(i)}$).
\end{enumerate}

We now describe each stage in detail. Our algorithm implements the greedy removal idea of Theorem 2 in~\cite{HM07},
so $j_i$ is selected as to minimize $\FNormS{\matX^{\dagger}_{{\cal S}_{i}}}$ subject to constraint ${\cal S}_{i}$ is obtained from ${\cal S}_{i-1}$ by removing a single entry, and that the rank of $\matX_{{\cal S}_{i}}$ is equal to the rank of $\matX_{{\cal S}_{i-1}}$. Specifically, the formula for $j_i$ is
\begin{eqnarray}
\label{eq:unstable}
j_i = \arg \min_{r \in {{\cal S}_{i-1}}; \TNorm{\y^{(i-1)}_r}<1}
\left(  \frac{ \sum_{l=1}^{n} \left( \y^{(i-1)}_{rl} /\sigma^{(i-1)}_l \right)^2 } {1-\TNormS{\y^{(i-1)}_r}} \right)\,.
\end{eqnarray}
In the last equation, $\y^{(i-1)}_{rl}$ is the the $l$ element of $\y^{(i-1)}_r$ or equivalently, $(l,r)$ element of $\matY^{(i-1)}$.
We will prove shortly that indeed $j_i$ is the minimizer we seek.

As for the second stage, we simply downdate the SVD of $\matX_{{\cal S}_{i-1}}$ to obtain an SVD of $\matX_{{\cal S}_{i}}$, using the algorithm described in~\cite{GE95}.

\begin{algorithm}[t]
\textbf{Input:} $\matX\in\R^{n \times m}$ ($m> n$, $\rank(\matX)=n$), sampling parameter \math{n \le k \le m}. \\
\noindent \textbf{Output:} Set $\cal S$ $\subseteq [m]$ of cardinality $k$.
\begin{algorithmic}[1]
\STATE ${\cal S}_0 \gets [m]$
\STATE Compute the SVD of $\matX_{{\cal S}_0} $:  $\matX_{{\cal S}_0} = \matU^{(0)} \matSig^{(0)} \matY^{(0)}$
\FOR {$i=1,2,\dots,m-k$}
\STATE Let the singular values of $\matX_{{\cal S}_{i-1}}$ be $\sigma^{(i-1)}_1,\ldots,\sigma^{(i-1)}_n$.
\STATE Let the columns of $\matY^{(i-1)}$ be $\{\y^{(i-1)}_r\}_{r \in {{\cal S}_{i-1}}}$.\\
Denote by $\y^{(i-1)}_{rl}$ the $l$-th element of $\y^{(i-1)}_r$.
\STATE $j_i \gets \arg \min_{r \in {{\cal S}_{i-1}}; \TNorm{\y^{(i-1)}_r}<1} \left(  \frac{ \sum_{l=1}^{n} \left( \y^{(i-1)}_{rl} /\sigma^{(i-1)}_l \right)^2 } {1-\TNormS{\y^{(i-1)}_r}} \right)\,.$ \\
\COMMENT{See proof on how to implement this step in a stable manner.}
\STATE ${\cal S}_{i} \gets {\cal S}_{i-1} - \{j_i\}$
\STATE Downdate the SVD of $\matX_{{\cal S}_{i-1}}$ to obtain an SVD of $\matX_{{\cal S}_{i}}=\matU^{(i)} \matSig^{(i)} \matY^{(i)}$.\\
\COMMENT{Using an algorithm described in~\cite{GE95}}
\ENDFOR
\RETURN $\cal S$
\end{algorithmic}
\caption{A deterministic greedy removal algorithm for subset selection (Theorem~\ref{thm1}).}
\label{alg1}
\end{algorithm}

Before proceeding to the analysis of Algorithm~\ref{alg1}, we discuss two numerical stability issues that affect an actual
implementation of Algorithm~\ref{alg1}.
Computing the terms in equation~\eqref{eq:unstable} might be problematic since the computation can potentially suffer from catastrophic cancellations when $\TNorm{\y^{(i-1)}_r} \approx 1$. However, to find the minimizer we need to do only comparisons. That is, we need to be able to determine for two indices $g,h\in {{\cal S}_{i-1}}$ whether
$$\frac{ \sum_{l=1}^{n} \left( \y^{(i-1)}_{gl} /\sigma^{(i-1)}_l \right)^2 } {1-\TNormS{\y^{(i-1)}_g}} \leq \frac{ \sum_{l=1}^{n} \left( \y^{(i-1)}_{hl} /\sigma^{(i-1)}_l \right)^2 } {1-\TNormS{\y^{(i-1)}_h}},$$
or not. It is easy to verify that provided $\TNorm{\y^{(i-1)}_g} < 1$ and $\TNorm{\y^{(i-1)}_h} < 1$,  the last equation holds if and only if
$$ \sum_{l=1}^{n} \left( \y^{(i-1)}_{gl} /\sigma^{(i-1)}_l \right)^2 + \TNormS{\y^{(i-1)}_g} \cdot \sum_{l=1}^{n} \left( \y^{(i-1)}_{hl} /\sigma^{(i-1)}_l \right)^2 \leq $$ $$
\sum_{l=1}^{n} \left( \y^{(i-1)}_{hl} /\sigma^{(i-1)}_l \right)^2 + \TNormS{\y^{(i-1)}_h} \cdot \sum_{l=1}^{n} \left( \y^{(i-1)}_{gl} /\sigma^{(i-1)}_l \right)^2\,.$$
The last equation does not do any subtraction, so it does not suffer from catastrophic cancellations.

Another issue with equation~\eqref{eq:unstable} is that an index $h \in {{\cal S}_{i-1}}$ is a candidate minimizer only if $\TNorm{\y^{(i-1)}_h} < 1$. Under inexact arithmetic that will always be the case, even if removing the column results in a rank deficient system. This issue can be solved by replacing the test $\TNorm{\y^{(i-1)}_h} < 1$ with $\TNorm{\y^{(i-1)}_h} < 1 - \tau$ for some small threshold $\tau$.

\begin{theorem}\label{thm1}
Fix $\matX \in \R^{n \times m}$ ($m>n$, $\rank(\matX)=n$) and sampling parameter $m \geq k \ge n$.
Algorithm~\ref{alg1} needs $O\left( mn^2 + mn\left(m - k\right) \right)$ operations and
deterministically constructs a set $\cal S$ $\subseteq [m]$ of cardinality $k$ with
\eqan{
\FNormS{  \matX_{ \cal S}^{\dagger} } \le \frac{m-n+1}{k-n+1} \cdot \FNormS{\matX^{\dagger}}
\qquad\mbox{and} \hspace{0.3in}
 \TNormS{ \matX_{ \cal S}^{\dagger} } \le \frac{m-n+1}{k-n+1} \cdot n \cdot \TNormS{\matX^{\dagger}}.
}
Moreover, if $\matX$ contains orthonormal rows, the operation count is $O\left(mn\left(m - k\right) \right)$.
\end{theorem}

Before proceeding with the proof we state an auxiliary lemma. However, we defer the proof to Section~\ref{sec:volume} since this Lemma is a corollary of a Theorem that appears in that section.

\begin{lemma}\label{lem:oneremove}
Let $\matX\in\mathbb{R}^{n\times m}$ ($m \ge n)$ be a full rank matrix. There exists a subset ${\cal S}\subset[m]$ of cardinality $m-1$ such that $\matX_{\cal S}$ is full rank and
$$
\FNormS{\matX^{\dagger}_{\cal S}}\leq \frac{m - n + 1}{m - n} \cdot \FNormS{\matX^{\dagger}}\,.
$$
\end{lemma}

\begin{proof}[Proof of Theorem~\ref{thm1}]
The spectral norm bound  is immediate from the Frobenius norm bound using the fact that for any
matrix $\matB$, 
$$\TNormS{\matB} \le \FNormS{\matB} \le \rank(\matB) \cdot \TNormS{\matB}.$$ So, we prove only the Frobenius norm bound.

We now prove that $j_i$, given by equation~\eqref{eq:unstable} (and line 6 in Algorithm~\ref{alg1})),  minimizes $\FNormS{\matX^{\dagger}_{{\cal S}_{i}}}$ subject to constraint ${\cal S}_{i}$ is obtained from ${\cal S}_{i-1}$ by removing a single entry, and that the rank of $\matX_{{\cal S}_{i}}$ is equal to the rank of $\matX_{{\cal S}_{i-1}}$.

First, we argue that for any $r\in{\cal S}_{i-1}$ the matrix 
$$\matX_{{\cal S}_{i-1}-\{r\}}\matX\transp_{{\cal S}_{i-1}-\{r\}}=\matX_{{\cal S}_{i-1}}\matX_{{\cal S}_{i-1}}\transp  - \x^{(i-1)}_r (\x^{(i-1)}_r)\transp$$ is singular if and only if $\TNorm{\y^{(i-1)}_r}=1$ (under the assumption that $\matX_{{\cal S}_{i-1}}\matX_{{\cal S}_{i-1}}\transp$ is non-singular).
Notice that,
$$\matX_{{\cal S}_{i-1}}\matX_{{\cal S}_{i-1}}\transp  - \x^{(i-1)}_r (\x^{(i-1)}_r)\transp = \matU^{(i-1)} \matSig^{(i-1)} \left(\matI_n - \y^{(i-1)}_r (\y^{(i-1)}_r)\transp \right)  \matSig^{(i-1)} (\matU^{(i-1)})\transp\,.$$
The matrix $\matU^{(i-1)}$ is full rank (it is square unitary), so we find that
$$\matX_{{\cal S}_{i-1}}\matX_{{\cal S}_{i-1}}\transp  - \x^{(i-1)}_r (\x^{(i-1)}_r)\transp$$ is singular if and only if
$$\matSig^{(i-1)} \left(\matI_n - \y^{(i-1)}_r (\y^{(i-1)}_r)\transp \right)  \matSig^{(i-1)}$$ is singular. We now observe that $\matSig^{(i-1)}$ is full rank (it is diagonal with positive values on the diagonal) as well, so we find that
$$\matX_{{\cal S}_{i-1}}\matX_{{\cal S}_{i-1}}\transp  - \x^{(i-1)}_r (\x^{(i-1)}_r)\transp$$ is singular if and only if
$$\matI_n - \y^{(i-1)}_r (\y^{(i-1)}_r)\transp$$ is singular. That can hold only if $\TNorm{\y^{(i-1)}_r} = 1$.  Therefore, comparing the norm of $\y^{(i-1)}_r$ with $1$ is an efficient way (once we have an SVD) under exact arithmetic to detect if 
$$\matX_{{\cal S}_{i-1}-\{r\}}\matX\transp_{{\cal S}_{i-1}-\{r\}}$$ 
is singular. This justifies the restriction $\TNorm{\y^{(i-1)}_r}<1$ in equation~\eqref{eq:unstable}.

We proceed with some calculations.
Fix an index $r\in{\cal S}_{i-1}$. If $$\matX_{{\cal S}_{i-1}}\matX_{{\cal S}_{i-1}}\transp  - \x^{(i-1)}_r (\x^{(i-1)}_r)\transp$$ is not singular, then,
$\trace\left(  \left( \matX_{{\cal S}_{i-1}}\matX_{{\cal S}_{i-1}}\transp  - \x^{(i-1)}_r (\x^{(i-1)}_r)\transp  \right)^{-1}  \right)=$
\eqan{
\label{eq:dehoog-step}
&\buildrel{(a)}\over{=}&
\trace\left( \matU^{(i-1)} \left( (\matSig^{(i-1)})^2 -  \matSig^{(i-1)} \y^{(i-1)}_r (\y^{(i-1)}_r)\transp \matSig^{(i-1)} \right)^{-1} (\matU^{(i-1)})\transp \right) \\
&\buildrel{(b)}\over{=}&
\trace\left( (\matSig^{(i-1)})^{-2} + \frac{(\matSig^{(i-1)})^{-1} \y^{(i-1)}_r (\y^{(i-1)}_r)\transp
(\matSig^{(i-1)})^{-1}}{1 -  (\y^{(i-1)}_r)\transp \y^{(i-1)}_r }  \right) \\
&\buildrel{(c)}\over{=}&
\trace\left( (\matSig^{(i-1)})^{-2}\right) + \trace\left(\frac{(\matSig^{(i-1)})^{-1} \y^{(i-1)}_r (\y^{(i-1)}_r)\transp
(\matSig^{(i-1)})^{-1}}{1 -  (\y^{(i-1)}_r)\transp \y^{(i-1)}_r } \right)\\
&\buildrel{(d)}\over{=}&
\FNormS{\matX_{{\cal S}_{i-1}}^{\dagger}}  + \frac{\trace\left( (\matSig^{(i-1)})^{-1} \y^{(i-1)}_r \left( (\matSig^{(i-1)})^{-1} \y^{(i-1)}_r \right)\transp \right)}{1-\TNormS{\y^{(i-1)}_r}}  \\
&\buildrel{(e)}\over{=}&
\FNormS{\matX_{{\cal S}_{i-1}}^{\dagger}}  + \frac{ \TNormS{ (\matSig^{(i-1)})^{-1} \y^{(i-1)}_r} }{1-\TNormS{\y^{(i-1)}_r}} \\
& \buildrel{(f)}\over{=} &
\FNormS{\matX_{{\cal S}_{i-1}}^{\dagger}} + \frac{ \sum_{l=1}^{n} \left( \y^{(i-1)}_{rl} /\sigma^{(i-1)}_l \right)^2 } {1-\TNormS{\y^{(i-1)}_r}}
}
\math{(a)} follows by replacing the SVD of $\matX$ and the identity $(\matU \matA \matU\transp)^{-1} = \matU \matA^{-1} \matU\transp$ for a unitary $\matU$ and non-singular $\matA$.
\math{(b)} follows from the Sherman-Morrison formula (recall that we assume that the matrix is not singular, so $\TNorm{\y^{(i-1)}_r} \neq 1$) and that for any unitary $\matU$ we have $\trace(\matU \matA \matU\transp)=\trace(\matA)$.
\math{(c)} follows by the linearity of the trace operator.
\math{(d)} follows from the fact that $ 1/(1-\TNormS{\y^{(i-1)}_r})$ is a scalar.
\math{(e)} follows from the fact that for any matrix $\matB$ we have $\trace\left(\matB\matB\transp \right) = \FNormS{\matB}$; in our case, we apply this equality to $\matB = (\matSig^{(i-1)})^{-1} \y^{(i-1)}_r$.
\math{(f)} follows because $\matSig^{(i-1)}$ is diagonal.

These calculations, alongside the observation that 
$$\matX_{{\cal S}_{i-1}}\matX_{{\cal S}_{i-1}}\transp  - \x^{(i-1)}_r (\x^{(i-1)}_r)\transp$$ is singular if and only if $\TNorm{\y^{(i-1)}_r} = 1$, imply that indeed using equation~\eqref{eq:unstable}
we can find the $j_i$ such that $\FNormS{\matX^{\dagger}_{{\cal S}_{i}}}$ is minimized.

We now use this fact to establish the Frobenius norm approximation bound (recall that the spectral norm bound follows immediatly from the Frobenius norm bound). We will show using induction that
\begin{equation}
\FNormS{  \matX_{{\cal S}_i}^{\dagger} } \le \frac{m-n+1}{m-i-n+1} \cdot \FNormS{\matX^{\dagger}}\,.\label{eq:induct}
\end{equation}
Since our algorithm returns ${\cal S}_{m-k}$ the claim follows from~\eqref{eq:induct}.

Equation~\eqref{eq:induct} trivially holds for $i=0$. Assume it holds for $i-1$. We now show it holds for $i$. Note that the cardinality of ${\cal S}_{i-1}$ is $m-i+1$. Lemma~\ref{lem:oneremove} ensures that there exists a subset ${\cal T}_{i} \subset {\cal S}_{i-1}$ of cardinality $m-i$ such that
\eqan{
\FNormS{  \matX_{{\cal T}_{i}}^{\dagger} } \le \frac{m-i-n+2}{m-i-n + 1} \cdot \FNormS{  \matX_{{\cal S}_{i}}^{\dagger} }
&\leq& \frac{m-i-n+2}{m-i-n + 1} \cdot \frac{m-n+1}{m-i-n+2} \cdot \FNormS{\matX^{\dagger}}\\
&=&
\frac{m-n+1}{m-i-n+1} \cdot \FNormS{\matX^{\dagger}}\,.
}
Our algorithm finds a subset ${\cal S}_{i} \subset {\cal S}_{i-1}$ of cardinality $m-i$ with minimal $\FNormS{  \matX_{{\cal S}_{i}}^{\dagger} }$, so
$$\FNormS{  \matX_{{\cal S}_{i}}^{\dagger} } \leq \FNormS{  \matX_{{\cal T}_{i} }^{\dagger} } \leq \frac{m-n+1}{m-i-n+1} \cdot \FNormS{\matX^{\dagger}}\,.$$

We conclude by analyzing the operation count. At the start, Algorithm~\ref{alg1} requires $O(mn^2)$ operations to compute a thin SVD of $\matX$.
Now, at iteration $i$, computing $j_i$ requires $O((m-i+1)n)$ operations. Downdating the SVD to find
$\matSig^{(i)}$ and $\matY^{(i)}$
can be done be done in $O((m-i+1)n)$ operations \footnote{This is precisely Problem 3 in page 794 of~\cite{GE95}; the third paragraph in page 795 of~\cite{GE95}
argues that this problem can be solved in $O(mn \log^2 \epsilon)$ operations, where $\epsilon$ is the machine precision. In our analysis we ignore the $\log^2 \epsilon$ term since $\epsilon$ is constant, and $\log^2 \epsilon$ is not too big since typically $\epsilon \approx 10^{-16}$. Ignoring such terms is common in the analysis of SVD-type algorithms.
}. There are $m-k$ iterations, so overall, $O\left( mn^2 + mn\left(m - k\right) \right)$ operations suffice. If $\matX$ has orthonormal rows, the operation count is just
$O\left( mn\left(m - k\right) \right)$ because the initial SVD is available.
\end{proof}

\subsection{Deterministic Greedy Removal (spectral norm)}\label{sec:greedy_removal2}

We now describe an algorithm which achieves a slightly worse Frobenius norm bound than the bound in the previous theorem
but a slightly better spectral norm bound. Algorithm~\ref{alg2} is the pseudo-code description. Algorithm~\ref{alg2} simply
applies Algorithm~\ref{alg1} on $\matV\transp$, where $\matV \in\R^{m\times n}$  is the matrix containing the top $n$ right singular vectors of $\matX$.  Notice that the output of Algorithm~\ref{alg1} and Algorithm~\ref{alg2}  might be different.
\begin{algorithm}[t]
\textbf{Input:} $\matX\in\R^{n \times m}$ ($m> n$, $\rank(\matX)=n$), sampling parameter \math{n \le k \le m}. \\
\noindent \textbf{Output:} Set $\cal S$ $\subseteq [m]$ of cardinality $k$.
\begin{algorithmic}[1]

\STATE Compute the matrix $\matV \in\R^{m\times n}$  of the right singular vectors corresponding to the top $n$ singular values of $\matX$.

\STATE Run Algorithm~\ref{alg1} with inputs $\matV\transp$ and $k$ to obtain $\cal S$ of cardinality $k$.

\RETURN ${\cal S}$

\end{algorithmic}
\caption{A deterministic greedy removal algorithm for subset selection (Corollary~\ref{cor1}).}
\label{alg2}
\end{algorithm}

\begin{corollary}\label{cor1}
Fix $\matX \in \R^{n \times m}$ ($m>n$, $\rank(\matX)=n$) and sampling parameter $k \ge n$. Algorithm~\ref{alg2} needs $O\left( mn^2 + mn\left(m - k\right) \right)$ operations
and deterministically constructs a set $\cal S$ $\subseteq [m]$ of cardinality $k$ with
$$\FNormS{  \matX_{ \cal S}^{\dagger} } \le \frac{m-n+1}{k-n+1} \cdot n \cdot \TNormS{\matX^{\dagger}}.$$
Also, for $i=1,\dots,n$ we have
$$\sigma_i^2(\matX) \cdot \left(1+ \frac{n\left(m-k\right)}{k-n+1} \right)^{-1} \le \sigma_i^2(\matX_{\cal S})\,.$$
In particular,
$$\TNormS{ \matX_{ \cal S}^{\dagger} } \le \left(1+ \frac{n\left(m-k\right)}{k-n+1} \right) \cdot \TNormS{\matX^{\dagger}}\,.$$
Moreover, if $\matX$ contains orthonormal rows, the operation count is $O\left(mn\left(m - k\right) \right)$.
\end{corollary}
\begin{proof}
Let ${\cal S} \subseteq [m]$ be the set found by the algorithm, and let
$\bar{\cal S} = [m]-{\cal S}$. Corollary 2 of~\cite{HM07} asserts that
$$ \FNormS{ \matX_{\cal S}^{\dagger} \cdot \matX_{\bar{\cal S}}} \le \frac{n(m-k)}{k-n+1}.$$
Now, Corollary 1 in~\cite{HM07} indicates that, if such a bound holds for $\FNormS{ \matX_{\cal S}^{\dagger} \cdot \matX_{\bar{\cal S}}}$,
then,
\eqan{
\FNormS{  \matX_{ \cal S}^{\dagger} } \le \frac{m-n+1}{k-n+1} \cdot n \cdot \TNormS{\matX^{\dagger}}; \hspace{-0.24in}
\qquad\mbox{for $i=1,...,n:$} 
\sigma_i^2(\matX) \cdot \left(1+ \frac{n\left(m-k\right)}{k-n+1} \right)^{-1} \le \sigma_i^2(\matX_{\cal S})\,.
}
\end{proof}

\subsection{Deterministic Greedy Selection}
The algorithm of this section builds the set ${\cal S}$ by iteratively adding columns to it, after starting 
with the empty set. It uses a deterministic algorithm presented in~\cite{BDM11a},
which is, in turn, a generalization of an algorithm from~\cite{BSS09}.
In particular, we use Lemma 10 from~\cite{BDM11a}.
\begin{lemma}[Dual Set Spectral Sparsification, Lemma 10 in~\cite{BDM11a}.] \label{lem:2set}
Let \math{\cl V=\{\v_1,\ldots,\v_m\}} and \math{\cl U=\{\u_1,\ldots,\u_m\}}
be two equal cardinality decompositions of identity matrices: \math{\v_i\in\R^{n}} ($n < m$), \math{\u_i\in\R^\ell} ($\ell \leq m$),
$\sum_{i=1}^m\v_i\v_i\transp=\matI_{n}$, and $\sum_{i=1}^m\u_i\u_i\transp=\matI_{\ell}$.
Given an integer \math{k} with \math{n < k \le m}, there exists an algorithm that computes a set of weights
\math{s_i\ge 0} ($i=1,\ldots,m$) {at most \math{k} of which are non-zero}, such that
\eqan{
\sigma_{n}\left(\sum_{i=1}^m s_i\v_i\v_i\transp\right)
\ge
\left(1 - \sqrt{\frac{n}{k}}\right)^2
\qquad\mbox{and} \hspace{0.3in}
\sigma_{1}\left(\sum_{i=1}^m s_i\u_i\u_i\transp\right)
\le \left(1 + \sqrt{ \frac{\ell}{k} }\right)^2.
}
The algorithm is deterministic and needs at most $O\left(k m \left(n^2+\ell^2\right) \right)$ operations.
Moreover, if the set $\cl U$ contains vectors from the standard basis from $\R^{m}$, the algorithm needs
$O\left(k m n^2 \right)$ operations. We denote the application of the algorithm to $\cl V$ and $\cal U$ by
$$[s_1, s_2, \dots, s_m] = \textsc{DualSet}( \cl V, \cl U, k ).$$
\end{lemma}

We refer the reader to~\cite{BDM11a} for the full description of the algorithm. Lemma~\ref{lem:2set}
implies that one can sample from two different set of vectors \math{\cl V=\{\v_1,\ldots,\v_m\}} and \math{\cl U=\{\u_1,\ldots,\u_m\}}, and control \emph{simultaneously} the smallest singular value of the matrix formed from the sampled vectors from the first set, and the largest singular value of the matrix formed from the sampled vectors from the second set.

A complete pseudo-code description of our algorithm appears as Algorithm~\ref{alg3}. Algorithm~\ref{alg3} proceeds as follows.
First, it computes the SVD of $\matX$: $\matX = \matU \matSig\matV\transp$ (see also Section~\ref{sec:pre} for useful notation).
The second step is to apply the algorithm of Lemma~\ref{lem:2set} ($\textsc{DualSet}$) on
${\cal V} = \{ \y_1, \dots, \y_m \}$
and
${\cal U} = \{ \e_1, \dots, \e_m \},$ the standard basis,
to compute the weights $s_1,\dots,s_m$. The algorithm then returns the set of non-zero $s_i$'s:
${\cal S} = \{ i\,:\,s_i \neq 0\}$.

We now present the analysis of Algorithm~\ref{alg3}.

\begin{algorithm}[t]
\textbf{Input:} $\matX\in\R^{n \times m}$ ($m> n$, $\rank(\matX)=n$), sampling parameter \math{n \le k \le m}. \\
\noindent \textbf{Output:} Set $\cal S$ $\subseteq [m]$ of cardinality at most $k$.
\begin{algorithmic}[1]

\STATE Compute the matrix $\matV \in\R^{m\times n}$  of the top $n$ right singular vectors of $\matX$.

\STATE Let ${\cal V} = \{ \y_1, \dots, \y_m \}$  (see Section~\ref{sec:pre} for the definition of $\y_i$'s).

\STATE Let ${\cal U} = \{ \e_1, \dots, \e_m \}$ contain the standard basis vectors.

\STATE Run $[s_1, s_2, \dots, s_m] = \textsc{DualSet}( \cl V, \cl U, k )$.

\RETURN ${\cal S} = \{ i\,:\,s_i \neq 0\}$

\end{algorithmic}
\caption{A deterministic greedy selection algorithm for subset selection (Theorem~\ref{thm2}.)}
\label{alg3}
\end{algorithm}

\begin{theorem}\label{thm2}
Fix $\matX \in \R^{n \times m}$ ($m>n$, $\rank(\matX)=n$) and sampling parameter $m \geq k > n$.
Algorithm~\ref{alg3} needs $O\left( k m n^2 \right)$ operations and
deterministically constructs a set $\cal S$ $\subseteq [m]$ of cardinality at most $k$ such that for both $\xi=2, \mathrm{F}$:
$$ \XNormS{ \matX_{\cal S}^{\dagger} } \le \left( 1 + \sqrt{\frac{m}{k}}\right)^2 \left( 1 - \sqrt{\frac{n}{k}} \right)^{-2} \XNormS{ \matX^{\dagger} }.$$
\end{theorem}

\begin{proof}
We first prove the approximation bound, and then bound the number of operations.

Lemma~\ref{lem:2set} guarantees that
$$\sigma_{1}\left(\sum_{i=1}^m s_i\e_i\e_i\transp\right)
\le \left(1 + \sqrt{ \frac{m}{k} }\right)^2\,.
$$
However, $\sum_{i=1}^m s_i\e_i\e_i\transp = \diag(s_1,\dots,s_m) \in \R^{m \times m}$, a diagonal matrix containing the weights $s_i$'s
in its main diagonal;
so, $\max_i s_i \leq \left(1 + \sqrt{ \frac{m}{k} }\right)^2$.
Lemma~\ref{lem:2set} also guarantees that
$$\sigma_{n}\left(\sum_{i=1}^m s_i\y_i\y_i\transp\right)
\ge
\left(1 - \sqrt{\frac{n}{k}}\right)^2\,.
$$
Assume that ${\cal S} = \{ i_1, \dots, i_{\tilde{k}} \}$ where $\tilde{k} \leq k$ and $i_1 < i_2 < \dots < i_{\tilde{k}}$, and let $\matD = \diag(\sqrt{s_{i_1}},\dots,\sqrt{s_{i_{\tilde{k}}}})$.
It is easy to verify that $\sum_{i=1}^m s_i\y_i\y_i\transp = \matY_{\cal S} \matD^2 \matY\transp_{\cal S}$; so, $\matY_{\cal S}$ is full rank and $\TNormS{(\matY_{\cal S} \matD)^{\dagger}} \leq \left(1 - \sqrt{ \frac{n}{k} } \right)^{-2}$.
The bound $\max_i s_i \leq \left(1 + \sqrt{ \frac{m}{k} }\right)^2$ earlier implies that $\TNormS{\matD} \leq \left(1 + \sqrt{ \frac{m}{k} }\right)^2$.
Now, observe that,
\eqan{
\XNormS{ \matX^{\dagger}_{\cal S}} \buildrel{(a)}\over{=} \XNormS{\left(\matU \matSig \matY_{\cal S}\right)^{\dagger}}
\buildrel{(b)}\over{=}
\XNormS{ \matY_{\cal S}^{\dagger} \matSig^{-1} \matU\transp}
&\buildrel{(c)}\over{\leq}&
\TNormS{ \matY_{\cal S}^{\dagger}} \cdot \XNormS{ \matX^{\dagger}}\\
&\buildrel{(d)}\over{=}& 
\TNormS{ \left(\matY_{\cal S} \matD \matD^{-1} \right)^{\dagger}} \cdot \XNormS{ \matX^{\dagger}}\\
&\buildrel{(e)}\over{\le}&
\TNormS{ \matD } \cdot \TNormS{ \left(\matY_{\cal S} \matD \right)^{\dagger}} \cdot \XNormS{ \matX^{\dagger}}\\
&\buildrel{(f)}\over{\le}&
 \left(1 + \sqrt{\frac{m}{k}}\right)^2 \cdot \left(1 - \sqrt{ \frac{n}{k} } \right)^{-2} \cdot \XNormS{ \matX^{\dagger}} \\
}
\math{(a)} follows by replacing $\matX$ with its SVD.
\math{(b)} follows by using Lemma~\ref{lem:pseudo} and the fact that all three matrices involved are full rank.
\math{(c)} follows by standard properties of matrix norms, and using the definition of the pseudoinverse of $\matX$ and $\matSig$.
\math{(d)} follows by introducing the identity matrix $\matI_{\tilde{k}} = \matD \matD^{-1}$.
\math{(e)} follows by using Lemma~\ref{lem:thm2aux}. Finally, \math{(f)} follows from the bounds we just proved for the terms
$\TNormS{ \matD }$, and $\TNormS{ \left(\matY_{\cal S} \matD \right)^{\dagger}}$.

We conclude by analyzing the operation count. The algorithm first computes an SVD of $\matX$, which costs $O\left( mn^2 \right)$. The second step is to run the algorithm of Lemma~\ref{lem:2set} on the right singular vectors of $\matX$ and the standard basis, which costs  $O\left( k m n^2 \right)$. So the total cost is $O\left( k m n^2 \right)$.
\end{proof}

\subsection{Randomized Selection}

The main idea in the algorithm of this section is to non-uniformly sample columns from $\matX$. The analysis is based
on a matrix concentration bound from~\cite{RV07}.
More specifically, we use Theorem 3.1 from~\cite{RV07} (the constants are from Corollary 4 in~\cite{Vir10}).
\begin{lemma}[Theorem 3.1 in~\cite{RV07}]\label{lem:random}
Let $\x \in \R^n$ be a random vector, which is uniformly bounded almost everywhere: $\TNorm{\x} \le M$.
Assume, for normalization, that $\TNorm{ \Expect{ \x \x\transp } } \le 1.$ Let $\x_1, \x_2,...,\x_k$ be $k$
independent copies of $\x$ sampled with replacement. Then, for every $\epsilon \in (0,1)$, and with probability at
least $1 - 2 \cdot n \cdot e^{- \epsilon^2 k / 4 M^2}$: $ \TNorm{ \frac{1}{k} \sum_{i=1}^k \x_i \x_i\transp - \Expect{ \x \x\transp } } \le \epsilon.$
\end{lemma}

Our algorithm is based on non-uniform sampling of columns with replacement. This type of sampling is the basis of many randomized matrix algorithms~\cite{DKM06a,DMMS11}.
The sampling probabilities are related to the
so-called leverage scores of the columns of $\matX$~\cite{BMD09a, DMMS11}, but in our algorithm we make sure that no column has a sampling probability that is too small.  One needs cubic time to compute these probabilities
using SVD or QR; our algorithm computes the probabilities that way. However, one can approximate these
probabilities in sub-cubic time using recent results from~\cite{DMMW12}.
It might be the case that these results can be used to improve the running time of our algorithm, at the cost of some small increase in the approximation bound. However, we leave this issue for future research.

\begin{algorithm}[t]
\textbf{Input:} $\matX\in\R^{n \times m}$ ($m> n$, $\rank(\matX)=n$), sampling parameter \math{n \le k \le m}. \\
\noindent \textbf{Output:} Set $\cal S$ $\subseteq [m]$ of cardinality at most $k$.
\begin{algorithmic}[1]

\STATE Compute the matrix $\matV \in\R^{m\times n}$  of the top $n$ right singular vectors of $\matX$.
\STATE For $i=1,2,\dots,m$ let (see Section~\ref{sec:pre} for the definition of $\y_i$'s),
               $$ \tau_i = \max\{ \TNormS{\y_i}, \frac{n}{m}\}\, $$
and
               $$p_i = \tau_i / \sum_{j=1}^m \tau_j.$$
\FOR {$t=1,2,\dots,k$}
\STATE Pick $i_t$; where  $i_t = i$ with probability $p_i$.
\ENDFOR
\RETURN ${\cal S} = \{ i_1, i_2, \dots, i_t \}$

\end{algorithmic}
\caption{A randomized algorithm for subset selection (Theorem~\ref{thm3}.)}
\label{alg4}
\end{algorithm}

A complete pseudo-code description of our algorithm appears as Algorithm~\ref{alg4}. Algorithm~\ref{alg4} proceeds as follows. 
First, it computes the SVD of $\matX$: $\matX = \matU \matSig\matV\transp$ (see also Section~\ref{sec:pre} for useful notation).
Let, $$ \tau_i = \max\{ \TNormS{\y_i}, \frac{n}{m}\}, $$ for $i=1,\dots,m$.
The set ${\cal S}$ is formed by non-uniformly, and independently, sampling $k$ numbers from ${1,\dots,m}$ with replacement.
In each trial, $i$ is sampled with probability $$p_i = \tau_i / \sum_{j=1}^m \tau_j.$$

We now present the analysis of Algorithm~\ref{alg4}.

\begin{theorem}\label{thm3}
Fix $\matX \in \R^{n \times m}$ ($m> n$, $\rank(\matX)=n$). Choose
a probability parameter $\delta$ ($0 < \delta < 1$). Now, choose an sampling parameter
$m \geq k \geq \min(\ceil{32 n \ln (2n/\delta)},m)$.
Algorithm~\ref{alg4} needs $O\left( mn^2 + k \log k \right)$ operations and
randomly constructs a set $\cal S$ $\subseteq [m]$ with cardinality at most $k$,
such that, for both $\xi=2, \mathrm{F}$, and with probability at least $1 - \delta$,
$$ \XNormS{ \matX_{\cal S}^{\dagger} } \le
4 \cdot m \cdot \XNormS{ \matX^{\dagger} }.$$
Moreover, if $\matX$ contains orthonormal rows the operation count is $O\left( m n + k \log k \right)$.
\end{theorem}
\begin{proof}
We first prove the approximation bound, and then bound the number of operations.

The first part of the approximation bound analysis is a technical manipulation to enable us to use Lemma~\ref{lem:random}. First, let $c_1,\dots,c_k$ be the indices sampled in trials $1,\dots,k$. That is ${\cal S} = \{ c_1, \dots, c_k \}$. For $i=1,\dots,k $ define the random vector $\x_i = \y_{c_i}/\sqrt{p_{c_i}}$. Now, for $i=1,\dots,m$ define $s_i = \frac{1}{kp_i} \cdot \#\{j\,:\,c_j=i\}$.
Notice that ${\cal S} = \{ i\,:\,s_i \neq 0\}$. Assume that ${\cal S} = \{ i_1, \dots, i_{\tilde{k}} \}$ where $\tilde{k} \leq k$ and $i_1 < i_2 < \dots < i_{\tilde{k}}$, and let 
$$\matD = \diag(\sqrt{s_{i_1}},\dots,\sqrt{s_{i_{\tilde{k}}}}).$$
With these definitions we observe that 
$$\frac{1}{k} \sum_{i=1}^k \x_i \x_i\transp = \matY_{\cal S} \matD^2 \matY\transp_{\cal S}.$$

Notice that $\x_1,\x_2,\dots,\x_k$ are i.i.d. Let $\x$ denote a random vector from the same distribution of $\x_1,\x_2,\dots,\x_k$. To use Lemma~\ref{lem:random} we need to compute $\Expect{ \x \x\transp}$ and to bound $\TNormS{\x}$:
$$\Expect{ \x \x\transp} = \sum_{i=1}^{m} p_i \cdot \frac{1}{\sqrt{p_i}} \y_i \cdot \frac{1}{\sqrt{p_i}} \y_i\transp = \sum_{i=1}^{m}\y_i\y_i\transp = \matI_n\,;$$
\eqan{
\TNormS{\x} \buildrel{(a)}\over{\le} \max_{j \in [m]}  \frac{\TNormS{\y_j}}{p_j}
\buildrel{(b)}\over{=}
\max_{j \in [m]} \frac{\sum_{i=1}^m \tau_i}{\tau_j} \TNormS{\y_j}
\buildrel{(c)}\over{\le}
\sum_{i=1}^m \tau_i
&\buildrel{(d)}\over{=}&
\sum_{i=1}^m \max\{ \TNormS{\y_i}, \frac{n}{m} \}\\
&\buildrel{(e)}\over{\le}&
n + \sum_{i=1}^m \TNormS{\y_i}
\buildrel{(f)}\over{=}
2 n
}
\math{(a)} follows by replacing the values taken by the vector $\x$.
\math{(b)} follows by replacing the value for the probabilities $p_i$'s.
\math{(c)} follows by the fact that $\tau_j \ge \TNormS{\y_j}$, for all $j=1,...,m$.
\math{(d)} follows by replacing the value for the parameters $\tau_i$'s.
\math{(e)} follows by simple algebra.

We are now ready to apply Lemma~\ref{lem:random} for the random vector $\y$ described above.
An immediate application of this Lemma ($M = \sqrt{2 n }$, $\epsilon=1/2$) and our bound on $k$
give that with probability at
least $1 - \delta$,
$$ \TNorm{ \matY_{\cal S}\matD^2 \matY_{\cal S}\transp - \matI_n } \le \frac{1}{2}.$$
(Recall that $\frac{1}{k} \sum_{i=1}^k \x_i \x_i\transp = \matY_{\cal S} \matD^2 \matY\transp_{\cal S}$). Standard matrix perturbation theory results~\cite{GV96} imply that for $i=1,...,n$
$$| \sigma_i^2\left(\matY_{\cal S}\matD\right)-1 | \leq \TNorm{\matY_{\cal S}\matD^2 \matY_{\cal S}\transp- \matI_n},$$
so, $i=n$ gives,
$$ \TNormS{(\matY_{\cal S} \matD)^{\dagger}} \le 2.$$
We bound $\TNormS{ \matD }$ as follows
\eqan{
\TNormS{ \matD } \buildrel{(a)}\over{=} \max_{j \in [m]} \matD_{jj}^2
\buildrel{(b)}\over{\le}
k \max_{j \in [m]} \left(\frac{1}{k p_j} \right)
&\buildrel{(c)}\over{\le}&
\max_{j \in [m]} \left(\frac{\sum_{i=1}^m \tau_i}{\tau_j} \right)\\
&\buildrel{(d)}\over{=}&
\left(\sum_{i=1}^m \tau_i\right) \cdot \max_{j \in [m]} \left(\frac{1}{\max\{ \TNormS{\y_j}, n/m \}} \right)  \\
&\buildrel{(e)}\over{\leq}&
2 \cdot n \cdot \max_{j \in [m]} \left(\frac{1}{\max\{ \TNormS{\y_j}, n/m \}} \right)\\
&\buildrel{(f)}\over{\le}&
2 \cdot n  \cdot \frac{m}{n} = 2 \cdot m
}
\math{(a)} follows because $\matD^2$ is a diagonal matrix.
\math{(b)} follows because each entry in $\matD^2$ might contain the term $1/kp_j$, at most $k$ times.
\math{(c)} follows by replacing the values for the probabilities.
\math{(d)} follows by replacing the values of the parameters $\tau_j$'s.
\math{(e)} follows by the fact that $$\sum_{i=1}^m \tau_i \le 2 n,$$ which we proved in Eqn. \math{(f)} in the previous calculations.
\math{(f)} follows by simple algebra.

To conclude the analysis of the approximation bound, notice that at the end of Theorem~\ref{thm2}, we implicitly proved that
$$ \XNormS{  \matX_{\cal S}^{\dagger} } \le \TNormS{ \matD } \cdot \TNormS{(\matY_{\cal S} \matD)^{\dagger}} \cdot \XNormS{ \matX^{\dagger}}.$$
Replace the bounds for $\TNormS{(\matY_{\cal S} \matD)^{\dagger}}$ and $\TNormS{\matD}$ in this bound to wrap up.

We conclude by analyzing the operation count. The algorithm first computes an SVD of $\matX$, which costs $O\left( mn^2 \right)$.
The probabilities can be calculated in $O(mn)$ and the sampling procedure can be implemented in $O(m + k \log k)$.
In total, the cost is $O\left( mn^2 + k \log k \right)$. If $\matX$ contains orthonormal rows,  $O(mn + k \log k)$ operations suffice.
\end{proof}

\subsection{\label{sec:volume}Volume based bounds and algorithms}
We now consider bounds and algorithms which construct the set $\cal S$
by looking at the volume of the parallelepiped spanned by the columns of $\matX_{\cal S}$,
which is exactly the determinant of $\matX_{\cal S} \matX_{\cal S}\transp $.

\subsubsection{Subset Selection and Determinants}
We start with Lemma~\ref{lem:square-to-det}, which establishes the connection between determinants and subset selection.
\begin{lemma}
\label{lem:square-to-det}Let $\matX \in\mathbb{R}^{n\times m}$ ($m \ge n)$
be a full rank matrix, and let ${\cal S} \subseteq [m]$ be any subset of cardinality
$k$ ($n \le k \le m$) such that $\matX_{\cal S}$ is full rank. For $i=1,\dots,n$, let $\matY_i \in \R^{(n-1) \times m}$ denote the matrix obtained after removing the $i$th row of $\matX$.
Then,
$$ \FNormS{\matX_{\cal S}^{\dagger}} = \frac{ \sum_{i=1}^{n} \det\left(  \left(\matY_i\right)_{\cal S} \left(\matY_i\right)\transp_{\cal S}  \right) }{\det\left( \matX_{\cal S} \matX_{\cal S}\transp \right) }.$$
Let $\matX = [\x_1, \x_2,...,\x_m]$ be the column representation of $\matX$. If $\cal S$ has cardinality exactly $n$, then
\[
\FNormS{\matX_{\cal S}^{-1}}
\le \TNormS{\matX^{\dagger}} \cdot \frac{\sum_{j=1}^{m}\sum_{i=1}^{n}\det\left(\matX_{\cal S}(i\rightarrow \x_j\right))^{2}}{\det\left(\matX_{\cal S}\right)^{2}}.
\]
Recall that, $\matX_{\cal S}\left(i\rightarrow \x_j\right)$ is the matrix by replacing the $i$-th
column of $\matX_{\cal S}$ with $\x_j$, the
$j$-th column of $\matX$.
If $\matX$ has orthonormal rows then the last inequality is an equality.
\end{lemma}
\begin{proof}
We first prove the equality in the Lemma ($\cal S$ has cardinality k unless otherwise stated),
\begin{eqnarray*} \FNormS{\matX_{\cal S}^{\dagger}}
\buildrel{(a)}\over{=}   \trace\left( \left( \matX_{\cal S} \matX_{\cal S}\transp \right)^{-1}  \right)
&\buildrel{(b)}\over{=}& \trace\left( \det\left(\matX_{\cal S} \matX_{\cal S}\transp  \right)^{-1} \Adj\left(\matX_{\cal S} \matX_{\cal S}\transp \right) \right)\\
&\buildrel{(c)}\over{=}&    \det\left(\matX_{\cal S} \matX_{\cal S}\transp  \right)^{-1} \trace\left( \Adj\left(\matX_{\cal S} \matX_{\cal S}\transp \right) \right)\\
&\buildrel{(d)}\over{=}&    \det\left(\matX_{\cal S} \matX_{\cal S}\transp  \right)^{-1} \sum_{i=1}^{n} \left(  \Adj\left(\matX_{\cal S} \matX_{\cal S}\transp \right) \right)_{ii} \\
&\buildrel{(e)}\over{=}&    \det\left(\matX_{\cal S} \matX_{\cal S}\transp  \right)^{-1} \sum_{i=1}^{n}  \det\left(  \left(\matY_i\right)_{\cal S} \left(\matY_i\right)\transp_{\cal S}  \right)
\end{eqnarray*}
\math{(a)} follows by a property which connects the Frobenius norm of the pseudoinverse with the trace operator.
\math{(b)} follows by the well known formula for the inverse of a matrix using the adjugate matrix.
\math{(c)} follows by the linearity of the trace operator.
\math{(d)} follows by the definition of the trace operator.
Finally, \math{(e)} follows by the definition of the adjugate matrix and the observation that the $i$th diagonal
element of $\Adj\left(\matX_{\cal S} \matX_{\cal S}\transp \right)$ equals the determinant of an $(n-1) \times (n-1)$
matrix which is exactly $\matX_{\cal S} \matX_{\cal S}\transp$ after removing its $i$th row and $i$th column. This matrix is exactly $\left(\matY_i\right)_{\cal S} \left(\matY_i \right)\transp_{\cal S}$.

We now prove the second inequality ($\cal S$ has now fixed cardinality $n$).
Let $\matX=\matU\matSig \matV\transp$ be an SVD of $\matX$ (see also Section~\ref{sec:pre} for useful notation).
Define $\x_j = \matU \matSig \y_j$.
Then,
\begin{eqnarray*}
\FNormS{\matX_{\cal S}^{-1}} \buildrel{(a)}\over{=} \FNormS{\matY_{\cal S}^{-1} \matSig^{-1}\matU\transp }
 & \buildrel{(b)}\over{\le} &           \TNormS{ \matSig^{-1} } \cdot \FNormS{\matY_{\cal S}^{-1} \matY} \\
 & \buildrel{(c)}\over{=}&  \TNormS{ \matSig^{-1} } \cdot  \sum_{j=1}^{m}\TNormS{\matY_{\cal S}^{-1} \y_j} \\
 & \buildrel{(d)}\over{=}&  \TNormS{ \matSig^{-1} }
 \cdot  \sum_{j=1}^{m}\TNormS{\matY_{\cal S}^{-1} \matSig^{-1} \matU\transp \matU \matSig \y_j} \\
 & \buildrel{(e)}\over{=} & \TNormS{ \matSig^{-1} }
 \cdot\frac{\sum_{j=1}^{m}\sum_{i=1}^{n} \det\left(\left( \matU \matSig\matY_{\cal S} \right)(i\rightarrow\matU \matSig \y_j)\right)^{2}}{
 \det( \matU \matSig\matY_{\cal S})^{2}}\\
 & \buildrel{(f)}\over{=} & \TNormS{ \matX^{\dagger}}  \cdot
 \frac{\sum_{j=1}^{m}\sum_{i=1}^{n}\det(\matX_{\cal S}(i\rightarrow \x_j))^{2}}{\det(\matX_{\cal S})^{2}}\,.
\end{eqnarray*}
\math{(a)} follows by replacing the SVD of $\matX_{\cal S}^{-1}$.
Notice that $\matX_{\cal S}=\matU\matSig \matY_{\cal S}$
and $\matX_{\cal S}^{-1}=\matY_{\cal S}^{-1} \matSig^{-1}\matU\transp$.
The latter equality holds because $\matY_{\cal S}^{-1}$ is a square full rank matrix,
which is immediate by the assumption that $\rank(\matX_{\cal S})=n$.
\math{(b)} follows by first using a property of matrix norms, and, then,
dropping the square orthonormal matrix $\matU\transp$ and
inserting the matrix $\matY$, which has orthonormal rows, to the Frobenius norm term.
\math{(c)} follows by the definition of the Frobenius norm.
\math{(d)} follows by inserting the identity matrix $\matSig^{-1} \matU\transp \matU \matSig = \matI_n$.
\math{(e)} follows by applying Cramer's rule to the linear system $\matA \x = \b$, with $\matA = \matU \matSig \matY_{\cal S}$ and $\b = \matU \matSig \y_j$.
Finally, \math{(f)} follows by replacing the appropriate values for $\matX_{\cal S}$, $\x_j$, and $\TNormS{ \matSig^{-1} }$.

Notice that if $\matX$ has orthonormal rows then $\matSig$ is the identity matrix, and \math{(b)} becomes an equality.
\end{proof}

\subsubsection{Random Subsets Chosen via Volume Sampling}
Lemma~\ref{lem:square-to-det} connects determinants and the term $\FNormS{\matX_{\cal S}^{\dagger}}$, for \emph{any} set $\cal S$ for which $\matX_{\cal S}$ has full rank. In the related work part of the introduction, we also stated various results for the specific set $\hat{\cal S} \subseteq [m]$ of cardinality $k \ge n$ that maximizes
$\det(\matX_{\cal T} \matX_{\cal T}\transp)$
over all possible ${\cal T}$'s of cardinality $k$.
Unfortunately, finding the maximum volume (determinant) subset is not only NP-hard~\cite{Packer02,CM09},
but also exponentially hard to approximate~\cite{Koutis06,CM10},
so, these results do not yield an efficient algorithm. We solve this issue using randomization.

\begin{lemma}\label{lem:volrand}
\label{lem:square-expected-by-volume}Let $\matX\in\mathbb{R}^{n\times m}$
($m \ge n)$ be a full rank matrix and let $m \geq k \ge n$. Suppose that ${\cal S} \sim \volsamp(\matX, k)$. Then,
\begin{align}
\label{eqn:exp-bound} \Expect{\FNormS{\matX_{\cal S}^{\dagger}}}\leq\frac{m-n+1}{k-n+1} \cdot \FNormS{\matX^{\dagger}}\,.
\end{align}
(If for every set ${\cal S} \in \binset{[m]}{n}$ the matrix $\matX_{\cal S}$ is full rank
then this bound becomes an equality). Also,
for $i=1,\dots,n$,
$$\Expect{\sigma^{-2}_i\left(\matX_{\cal S}\right)} \leq \left(1+\frac{n(m-k)}{k-n+1} \right) \cdot \sigma^{-2}_i\left(\matX\right)\,.$$
In particular,
$$\Expect{\TNormS{\matX_{\cal S}^{\dagger}}} \leq \left(1+\frac{n(m-k)}{k-n+1} \right) \cdot \TNormS{\matX^{\dagger}}\,.$$
\end{lemma}
\begin{proof}
For $i=1,\dots,n$, let $\matY_i \in \R^{(n-1) \times m}$ denote the matrix obtained after removing the $i$th row of $\matX$. Using the definition of expectation and the equality of Lemma~\ref{lem:square-to-det},
\begin{eqnarray*}
\Expect{\FNormS{\matX_{\cal S}^{\dagger}}} 
&=& \frac{\sum_{{\cal S}\in\binset{[m]}{k}}\det(\matX_{\cal S}\matX_{\cal S}\transp)\FNormS{\matX_{\cal S}^{\dagger}}}{\sum_{{\cal S}\in\binset{[m]}{k}}\det(\matX_{\cal S}\matX_{\cal S}\transp)}\\
&\buildrel{(*)}\over{=}& \frac{\sum_{{\cal S}\in\binset{[m]}{k},\rank(\matX_{\cal S})=n}\det(\matX_{\cal S}\matX_{\cal S}\transp)\FNormS{\matX_{\cal S}^{\dagger}}}{\sum_{{\cal S}\in\binset{[m]}{k}}\det(\matX_{\cal S}\matX_{\cal S}\transp)} \\
&=& \frac{\sum_{{\cal S}\in\binset{[m]}{k},\rank(\matX_{\cal S})=n}\sum_{i=1}^{n}\det((\matY_{i})_{\cal S}(\matY_{i})_{\cal S}\transp)}{\sum_{{\cal S}\in\binset{[m]}{k}}\det(\matX_{\cal S}\matX_{\cal S}\transp)}\,.
\end{eqnarray*}
In \math{(*)}, if $\rank(\matX_{\cal S}) \ne n$, then, $\det(\matX_{\cal S}\matX_{\cal S}\transp)=0$, so it
can be ignored in the sum. 
We will now analyze the numerator and the denominator of the last relation separately. We start with the denominator. We have
\begin{eqnarray*}
\sum_{{\cal S}\in\binset{[m]}{k}}\det(\matX_{\cal S}\matX_{\cal S}\transp)  
&\buildrel{(a)}\over{=} & \sum_{{\cal S}\in\binset{[m]}{k}}\sum_{{\cal T}\in\binset{\cal S}{n}}\det(\matX_{{\cal T}}\matX_{{\cal T}}\transp) \\
& \buildrel{(b)}\over{=} & \binnum{m-n}{k-n}\sum_{{\cal T}\in\binset{[m]}{n}}\det(\matX_{{\cal T}}\matX_{{\cal T}}\transp)\\
 & \buildrel{(c)}\over{=} & \binnum{m-n}{k-n}\prod_{i=1}^{n}\sigma_{i}^{2}\,,
\end{eqnarray*}
where $\sigma_{1}\ge \sigma_2\ge\dots\ge\sigma_{n}$ are the singular values of $\matX$.
\math{(a)} follows by applying the Cauchy-Binet formula.
\math{(b)} follows from observing that each set in $\binnum{[m]}{n}$
is repeated exactly $\binnum{m-n}{k-n}$ times in the sum.
\math{(c)} follows by applying the Cauchy-Binet formula again and
 the fact that for symmetric positive-definite
matrices the determinant is equal to the product of
the eigenvalues.

As for the numerator, we have
\begin{eqnarray*}
\sum_{{\cal S}\in\binset{[m]}{k},\rank(\matX_{\cal S})=n}\sum_{i=1}^{n}\det((\matY_{i})_{\cal S}(\matY_{i})_{\cal S}\transp)
& \buildrel{(a)}\over{\leq} &
\sum_{{\cal S}\in\binset{[m]}{k}}\sum_{i=1}^{n}\det((\matY_{i})_{\cal S}(\matY_{i})_{\cal S}\transp)\\
 & =  & \sum_{i=1}^{n}\sum_{{\cal S}\in\binset{[m]}{k}}\det((\matY_{i})_{\cal S}(\matY_{i})_{\cal S}\transp) \\
 & \buildrel{(b)}\over{=} & \sum_{i=1}^{n}\sum_{{\cal S}\in\binset{[m]}{k}}\sum_{{\cal T}\in\binset{\cal S}{n-1}}\det((\matY_{i})_{{\cal T}}(\matY_{i})_{{\cal T}}\transp)\\
 & \buildrel{(c)}\over{=} & \sum_{i=1}^{n}\binnum{m-n+1}{k-n+1}\sum_{{\cal T}\in\binset{[m]}{n-1}}\det((\matY_{i})_{{\cal T}}(\matY_{i})_{{\cal T}}\transp)\\
 & \buildrel{(d)}\over{=} & \binnum{m-n+1}{k-n+1}\sum_{i=1}^{n}\det(\matY_{i}\matY_{i}\transp)\\
 & \buildrel{(e)}\over{=} & \binnum{m-n+1}{k-n+1}\sum_{i=1}^{n}\prod_{j\neq i}\sigma_{i}^{2} \,.
\end{eqnarray*}
\math{(a)} follows because we are adding only positive terms in the sum.
\math{(b)} follows by applying the Cauchy-Binet formula.
\math{(c)} follows from observing that each set in $\binset{[m]}{n-1}$
is repeated exactly $\binnum{m-n+1}{k-n+1}$ times in the sum.
\math{(d)} follows by applying the Cauchy-Binet formula again.
Finally, in \math{(e)}, the matrices $\matY_{i}\matY_{i}\transp$ are equal to the matrix obtained by
deleting the $i$-th column and the $i$-th row of $\matX\matX\transp$, so according to Lemma~\ref{GK12-lem6},
$$
\sum_{i=1}^{n}\det(\matY_{i}\matY_{i}\transp)=\sum_{i=1}^{n}\prod_{j\neq i}\sigma_{i}^{2}.
$$

We now conclude the first part of the proof as follows,
\begin{eqnarray*}
\Expect{\FNormS{\matX_{\cal S}^{\dagger}}}  \leq \frac{\binnum{m-n+1}{k-n+1}\sum_{i=1}^{n}\prod_{j\neq i}\sigma_{i}^{2}}{\binnum{m-n}{k-n}\prod_{i=1}^{n}\sigma_{i}^{2}}
  = \frac{\binnum{m-n+1}{k-n+1}\FNormS{\matX^{\dagger} } }{\binnum{m-n}{k-n}}
  =  \frac{m-n+1}{k-n+1} \cdot \FNormS{\matX^{\dagger} }\,.
\end{eqnarray*}
(If $\matX_{\cal S}$ is full rank for every ${\cal S} \in \binset{[m]}{n}$ then
$(a)$ in the previous calculations is an equality.)

We now prove the bounds for the singular values of $\matX_{\cal S}$.
Let ${\cal T}$ be any subset of $[m]$ of cardinality $k$ such that $\matX_{\cal T}$ has full rank, and let $\bar{\cal T} = [m] - {\cal T}$.
Notice that $\bar{\cal T}$ has cardinality $m-k$.
Let,
$$
\matW = \left(\begin{array}{cc}
             \matI_k & \matX_{\cal T}^{\dagger}\matX_{\bar{\cal T}} \\
             {\bf 0}_{(m-k) \times k} & \matI_{m-k}
          \end{array}
    \right)\in \R^{m \times m}\,.
$$
(Note that $\matX_{\cal T}^{\dagger}\matX_{\bar{\cal T}}\in \R^{k \times (m-k)}$). Since $\matX_{\cal T}$ has full rank, we have
$$\left( \begin{array}{cc} \matX_{\cal T} & {\bf 0}_{n \times (m-k)} \end{array} \right) \matW =
\left( \begin{array}{cc} \matX_{\cal T} & \matX_{\bar{\cal T}} \end{array} \right) = \matX \matPi,
$$
where $\matPi \in \R^{m \times m}$ is an appropriate permutation matrix.
Clearly $\matW$ is non-singular (it is a triangular matrix with a non-zero diagonal), so for $i=1,\dots,n$,
$$
\sigma^{-2}_i \left( \matX_{\cal T} \right) =
    \sigma^{-2}_i \left( \left( \begin{array}{cc} \matX_{\cal T} & {\bf 0}_{n \times (m-k)} \end{array} \right) \right ) =
    \sigma^{-2}_i \left( \matX \matPi \matW^{-1} \right) \leq
    \TNormS{\matW} \cdot \sigma^{-2}_i \left( \matX \matPi \right) \le  \TNormS{\matW} \cdot \sigma^{-2}_i \left( \matX \right)  \,.
$$
In the above, the two inequalities are a simple application of Theorem 3.1 in~\cite{EI95}; we also used the fact that $\TNorm{\matPi} = 1$.
To bound $\TNormS{\matW}$ we observe that,
$$\TNormS{\matW} \le 1 + \TNormS{\matX_{\cal T}^{\dagger} \matX_{\bar{\cal T}}} \le 1 + \FNormS{\matX_{\cal T}^{\dagger} \matX_{\bar{\cal T}}}.$$
Now, if ${\cal S} \sim \volsamp(\matX, k)$ then only $\cal S$'s for which $\matX_{\cal S}$ is full rank have positive probability of being sampled. This implies that for $i=1,\dots,n$,
$$\Expect{\sigma^{-2}_i(\matX_{\cal S})} \leq \Expect{  \left(1+ \FNormS{\matX_{\cal S}^{\dagger}  \matX_{\bar{\cal S}}  } \right) }\cdot \sigma^{-2}_i(\matX)\,.
$$
The above bound is obtained as follows. Recall the two inequalities proved above for any ${\cal T}$:
$$\sigma^{-2}_i \left( \matX_{\cal T} \right) \le \TNormS{\matW} \cdot \sigma^{-2}_i \left( \matX \right),$$ and
$$\TNormS{\matW} \le 1 + \FNormS{\matX_{\cal T}^{\dagger} \matX_{\bar{\cal T}}}.$$ To the get the bound, combine these two
inequalities, apply the resulting inequality to ${\cal T} = {\cal S}$ (recall that ${\cal S}$ takes only values for which $\matX_{\cal S}$ is full rank), and take expectation on both sides.

We now bound $\Expect{  \left(1+ \FNormS{\matX_{\cal S}^{\dagger}  \matX_{\bar{\cal S}}  } \right) }$.
Let $\matX = \matU \matSig \matV\transp$ be an SVD of $\matX$ and let us denote $\matY = \matV\transp$.
Then, it is easy to verify that $\matX_{\cal S}^{\dagger} \matX_{\bar{\cal S}}=\matY_{\cal S}^{\dagger} \matY_{\bar{\cal S}}$.
To bound the expected value of $\FNormS{\matY_{\cal S}^{\dagger} \matY_{\bar{\cal S}}}$ we observe that,
$$
\matY_{\cal S}^{\dagger} \matY \matPi =
    \left( \begin{array}{cc} \matY_{\cal S}^{\dagger}\matY_{\cal S} & \matY_{\cal S}^{\dagger} \matY_{\bar{\cal S}} \end{array} \right)\,.
$$
This implies that 
$$\FNormS{\matY_{\cal S}^{\dagger} \matY \matPi} = \FNormS{\matY_{\cal S}^{\dagger}\matY_{\cal S}} + \FNormS{\matY_{\cal S}^{\dagger} \matY_{\bar{\cal S}}}.$$
Notice that $\matY_{\cal S}^{\dagger}\matY_{\cal S}$ is a projection;
so, $\FNormS{\matY_{\cal S}^{\dagger}\matY_{\cal S}} = n$. We now have $\FNormS{\matY_{\cal S}^{\dagger} \matY \matPi } = n +  \FNormS{\matY_{\cal S}^{\dagger} \matY_{\bar{\cal S}}}$. $\matY \matPi$ has orthonormal rows; so,
$\FNormS{\matY_{\cal S}^{\dagger} \matY \matPi } = \FNormS{\matY^{\dagger}_{\cal S}}$. So,
$\FNormS{\matY_{\cal S}^{\dagger} \matY_{\bar{\cal S}}} = \FNormS{\matY^{\dagger}_{\cal S}} - n.$
Since $\volsamp(\matX, k)=\volsamp(\matY, k)$, the Frobenius norm bound in the lemma guarantees that,
$$
\Expect{\FNormS{\matY^{\dagger}_{\cal S}}} \leq \frac{n(m - n + 1)}{k - n + 1}\,.
$$
Plugging that into the previously established equality $\FNormS{\matY_{\cal S}^{\dagger} \matY_{\bar{\cal S}}} = \FNormS{\matY^{\dagger}_{\cal S}} - n$, we find that,
$$
\Expect{\FNormS{\matY_{\cal S}^{\dagger} \matY_{\bar{\cal S}}}} \leq \frac{n(m-k)}{k-n+1}\,.
$$
This immediately gives a bound on $\Expect{  \left(1+ \FNormS{\matX_{\cal S}^{\dagger}  \matX_{\bar{\cal S}}  } \right) }$, which concludes the proof.
\end{proof}

We can now prove the following corollary, which was previously stated as Lemma~\ref{lem:oneremove}.
\begin{corollary} [Restatement of  Lemma~\ref{lem:oneremove}]
\label{cor:oneremove}
Let $\matX\in\mathbb{R}^{n\times m}$ ($m \ge n)$ be a full rank matrix. There exists a subset ${\cal S}\subset[m]$ of cardinality $m-1$ such that $\matX_{\cal S}$ is full rank and
$$
\FNormS{\matX^{\dagger}_{\cal S}}\leq \frac{m - n + 1}{m - n} \cdot \FNormS{\matX^{\dagger}}\,.
$$
\end{corollary}
\begin{proof}
Let ${\cal T} \sim \volsamp(\matX, m-1)$. According to Lemma~\ref{lem:volrand} we have
$$\Expect{\FNormS{\matX_{\cal T}^{\dagger}}}\leq\frac{m-n+1}{m-n} \cdot \FNormS{\matX^{\dagger}}\,.$$
The random variable $\FNormS{\matX_{\cal T}^{\dagger}}$ is discrete, so it must assume at least one value larger than the expectancy with non-zero probability. Let ${\cal S}$ be one such set, so
$$\FNormS{\matX^{\dagger}_{\cal S}}\leq \frac{m - n + 1}{m - n} \cdot \FNormS{\matX^{\dagger}}\,.$$
 $\matX_{\cal S}$ must be full rank since ${\cal T}$ assumes it with some non-zero probability, but the distribution $\volsamp(\matX, m-1)$ gives a zero probability to every subset ${\cal R}$ for which $\matX_{\cal R}$ is rank deficient.
\end{proof}

If there exists a subset ${\cal S}$ of columns of cardinality $k$ such that these columns are
linearly dependent then \eqref{eqn:exp-bound} might be a strict inequality.
For example, let $$\matY = \left( \begin{array} {cc} \matI_{n \times n} & {\bf 0}_{n \times (m - n)} \end{array} \right).$$
There is only one set of cardinality $n$ that has positive volume (i.e., the set of columns is full rank): ${\cal T} = [n]$. Since this is the only set with positive probability we have
$$\Expect{\FNormS{\matY_{\cal S}^{\dagger}}}=n=\FNormS{\matY^{\dagger}} < (m-n+1)\FNormS{\matY^{\dagger}}=(m-n+1)n\,.$$

\subsubsection{Volume Sampling Subset Selection}
To use Lemma~\ref{lem:volrand}
in an algorithm one needs a method to sample a subset from $\volsamp(\matX, k)$. Computing the probabilities for all $\binnum{m}{k}$ subsets
and sampling according to them is not efficient; there are too many such sets. However, this is not necessary, since one can \emph{simulate} volume sampling using polynomial number of operations using recent results from~\cite{DR10,GK12}. More precisely, using the algorithm \textsc{VolumeSample} from~\cite{GK12} we can sample a subset ${\cal S}$ from $\volsamp(\matX, n)$ using $O(n^{3}m)$ operations.
Current determinant-based sampling algorithms~\cite{DR10,GK12} can sample only $k=n$ columns from $\matX$.
We leave it as an open question whether one can efficiently sample from $\volsamp(\matX, k)$ for arbitrary $k \ge n$.

\begin{algorithm}[t]
\textbf{Input:} $\matX\in\R^{n \times m}$ ($m> n$, $\rank(\matX)=n$), parameter $\eta > 0$. \\
\noindent \textbf{Output:} Set $\cal S$ $\subseteq [m]$ of cardinality $n$.
\begin{algorithmic}[1]

\STATE Let $\alpha =  (1+\eta) \cdot \left( m-n+1 \right) \cdot \FNormS{ \matX^{\dagger} }$.

\REPEAT
\STATE Apply \textsc{VolumeSample} from~\cite{GK12} to sample a subset ${\cal S}$ from $\volsamp(\matX, n)$.
\UNTIL  $\FNormS{\matX_{\cal S}^{-1}}  \le  \alpha $

\RETURN ${\cal S}$

\end{algorithmic}
\caption{A randomized volume-based sampling algorithm for subset selection (Theorem~\ref{thm4}.)}
\label{alg5}
\end{algorithm}

A complete pseudo-code description of our algorithm appears as Algorithm~\ref{alg5}.
Algorithm~\ref{alg5} proceeds as follows.
It starts by using \textsc{VolumeSample} from~\cite{GK12} to sample a subset ${\cal S}$ from $\volsamp(\matX, n)$ using $O(n^{3}m)$ operations.
We then compute $\FNormS{\matX_{\cal S}^{-1}}$ (note that \textsc{VolumeSample} must return a full rank $\matX_{\cal S}$) using $O(n^{3})$ operations and compare it to
$(1+\eta) \cdot \left( m-n+1 \right) \cdot  \FNormS{ \matX^{\dagger} }$.
If it is smaller than the bound we return ${\cal S}$, otherwise we repeat this procedure until we
find a satisfactory ${\cal S}$.

We now present the analysis of Algorithm~\ref{alg5}.

\begin{theorem}\label{thm4}
Fix $\matX \in \R^{n \times m}$ ($m \ge n$, $\rank(\matX)=n$).
Choose a parameter $\eta > 0$.
Upon termination, Algorithm~\ref{alg5} outputs a set ${\cal S} \subseteq [m]$ of cardinality $n$
such that,
\eqan{
 \FNormS{ \matX_{\cal S}^{-1} } \le (1+\eta) \cdot \left( m-n+1 \right) \cdot  \FNormS{ \matX^{\dagger} };
\hspace{0.1in}
\TNormS{ \matX_{\cal S}^{-1} }
\le(1+\eta) \cdot \left( m-n+1 \right) \cdot n \cdot \TNormS{ \matX^{\dagger} }.
}
For every $0< \delta < 1$, the algorithm will terminate after
$O\left( m n^3 \log{(1/\delta)}/\log{(1+\eta)}\right)$ operations with probability of at least $1-\delta$.
\end{theorem}
\begin{proof}
We first prove the approximation bound, and then bound the number of operations.

We only need to prove the approximation bound for the Frobenius norm.
The bounds for the spectral norm follow by the fact that for any matrix $\matB$ of rank $n$:
$\TNormS{\matB} \le \FNormS{\matB} \le n \TNormS{\matB}$. Notice that we repeat step 3 $t=1,2,...$
times (see step 3 in the algorithm), constructing sets ${\cal S}_1$, ${\cal S}_2 ...$.
We stop only if we find a subset for which the approximation bound holds (the threshold $\alpha$ is exactly the
one that appears in the theorem statement), so the bounds are satisfied upon termination of the algorithm.

We now turn our attention to the number of operations.
Lemma~\ref{lem:square-expected-by-volume} indicates that if $\cal S$ is sampled from $\volsamp(\matX, n)$,
then,
$\Expect{\FNormS{ \matX_{\cal S}^{\dagger} }} \le \left( m-n+1 \right) \cdot  \FNormS{ \matX^{\dagger} }$.
For the first iteration $t=1$, by Markov's inequality, we find that, with probability at most $1/(1+\eta)$:
$\FNormS{ \matX_{{\cal S}_1}^{\dagger} } > (1+\eta) \left( m-n+1 \right) \cdot  \FNormS{ \matX^{\dagger} }$.
Therefore, for a finite number of iterations $\ell > 1$, the probability that all $t=1,...,\ell$, satisfy
$\FNormS{ \matX_{{\cal S}_t}^{\dagger} } > (1+\eta)  \left( m-n+1 \right) \cdot  \FNormS{ \matX^{\dagger} }$ is at most
$1/(1+\eta)^\ell$. So, for any $0< \delta < 1$
the probability that $\ceil{\log{(1/\delta)}/\log{(1+\eta)}}$ successive iterations fail is smaller than $\delta$.

Each iteration (line 3) takes $O(n^{3}m)$. Combining this with the analysis in the
previous paragraph reveals that
for any $0< \delta < 1$, Algorithm~\ref{alg5} will finish after $O\left( m n^3 \log{(1/\delta)}/\log{(1+\eta)}\right)$ operations with probability of at least $1-\delta$.
\end{proof}

\section{Lower Bounds}\label{sec:opt}
This section provides lower bounds for the subset selection problem of Definition~\ref{def:prob}.
By lower bounds, we mean that there exists a matrix $\matX \in \R^{n \times m}$ such that for every ${\cal S}$ of cardinality $k \ge n$,
for $\xi = 2$ or $\xi = \mathrm{F}$,
we have $\XNormS{\matX_{\cal S}^{\dagger}} \ge \gamma \XNormS{\matX^{\dagger}}$, for some value of $\gamma$ which we call \emph{lower bound}.

\subsection{Lower bound for the spectral norm version of the subset selection problem}
Recall that the problem of Definition~\ref{def:prob} is defined for both $\xi=2$ and $\xi=\mathrm{F}$. Here we focus on the $\xi=2$ case
and provide a lower bound of the form 
$$\TNormS{\matX_{\cal S}^{\dagger}} \ge \gamma \TNormS{\matX^{\dagger}}.$$ We first state two known
results that will be used in our proof.

\begin{proposition}[Theorem 17 in~\cite{BDM11a}]
Let $\matA = [\e_1+\alpha\e_2, \e_1+\alpha\e_3,\ldots, \e_1+\alpha\e_{m+1}] \in\R^{(m+1)\times m}$ for some $\alpha > 0$.
If $m > 2$, then, for every subset ${\cal S} \subseteq [m]$ of cardinality $k$, we have
$$ \TNormS{\matA - \matA_{\cal S} \matA_{\cal S}^{\dagger} \matA} = \frac{m+\alpha^2}{k+\alpha^2} \cdot \TNormS{\matA - \matA_{n}},$$  where $\matA_n \in\R^{(m+1)\times m}$ is the best rank $n$ approximation to $\matA$.
\end{proposition}


\begin{proposition}[Lemma 7 in~\cite{BDM11a}]
Let $\matW \in \R^{d \times m}$,  parameter $n < \rank(\matW)$,
and sampling parameter $n \le k \le  m$. Let $\matW = \matU \matSig \matV\transp$ be the SVD  of $\matW$.
Let $\matZ$ be the first $n$ rows of $\matV\transp$.
For every subset ${\cal S} \subseteq [m]$ of cardinality $k$ for which $\matZ_{\cal S}$
has full rank, we have
$$
\TNormS{\matW - \matW_{\cal S} \matW_{\cal S}^{\dagger} \matW}\le \TNormS{\matW-\matW_n} + \TNormS{\left(\matW-\matW_n\right)_{\cal S}\matZ^{\dagger}_{\cal S}}\le
\left( 1 + \TNormS{\matZ^{\dagger}_{\cal S}}\right) \cdot \TNormS{\matW-\matW_n}\,.
$$
Here, $\matW_n \in \R^{d \times m}$ is the best rank $n$ approximation to $\matW$.
\end{proposition}

\begin{theorem}[Spectral Norm]\label{lower:spectral}
For any $\alpha>0$, $n > 0$, $m > 2$ with $m > n$, and $k$ with $n \le k \le m$, there exists
a full rank $n \times m$ matrix $\matX$ such that, for any subset $\cal S$ $\subseteq [m]$
of cardinality $k \ge n$ with $\rank(\matX_{\cal S}) = n$,
$$ \TNormS{ \matX_{\cal S}^{\dagger} } \ge \left( \frac{m+\alpha^2}{k+\alpha^2} - 1 \right) \cdot \TNormS{\matX^{\dagger}}.$$
\end{theorem}
\begin{proof}
We construct the matrix $\matX$ as follows. Let $\matA = [\e_1+\alpha\e_2, \e_1+\alpha\e_3,\ldots, \e_1+\alpha\e_{m+1}] \in\R^{(m+1)\times m}$, and let $\matA = \matU \matSig \matV\transp$ be the SVD decomposition of $\matA$.
Let $\matX$ consist of the first $n$ rows of $\matV\transp$.
We prove the bound using Theorem 17 and Lemma 7 from~\cite{BDM11a} (see the two propositions above).

Applying Lemma 7 from~\cite{BDM11a} on $\matA$ and $\matX$ and combining it with the bound from Theorem 17 of~\cite{BDM11a} mentioned above, gives
$$
\TNormS{\matX^{\dagger}_{\cal S}}\ge \frac{\TNormS{\matA - \matA_{\cal S} \matA_{\cal S}^{\dagger} \matA}}{\TNormS{\matA-\matA_n}} - 1 =
\frac{m+\alpha^2}{k+\alpha^2} - 1 = \left(\frac{m+\alpha^2}{k+\alpha^2} - 1\right) \TNormS{\matX^{\dagger}}.
$$
\end{proof}

As $\alpha \rightarrow 0$, the  bound in the above theorem is $m/k - 1$. If $k = (1 + \Omega(1))n$ then
the upper bound of the deterministic algorithm of Theorem~\ref{thm2} asymptotically matches this lower bound.
The upper bounds of the algorithms in the Theorems~\ref{thm1},~\ref{thm3}, and~\ref{thm4} and Corollary~\ref{cor1} are - asymptotically - slightly worse. However, if $k=(1+o(1))n$ there is a gap between the lower bound and the best upper bound.


\subsection{Lower bound for the Frobenius norm version of the subset selection problem}
Recall that the problem of Definition~\ref{def:prob} is defined for both $\xi=2$ and $\xi=\mathrm{F}$. Here we focus on the $\xi=\mathrm{F}$ case
and provide a lower bound of the form $\FNormS{\matX_{\cal S}^{\dagger}} \ge \gamma \FNormS{\matX^{\dagger}}$.
We first state a known result that will be used in our proof.

\begin{proposition}[Theorem 19 in~\cite{BDM11a}]
Consider a block diagonal matrix $\matB \in \R^{d \times m}$:
a matrix $\matA \in \R^{d/n \times m/n}$ of the form that appears in the proof of Theorem~\ref{lower:spectral} is repeated $n$ times on $\matB$'s main diagonal.
For any $n \le \rank(\matB)$, and $k \ge n$,
any subset $\cal S$ of $k$ columns of $\matB$ satisfies,
$$ \FNormS{\matB - \matB_{\cal S} \matB_{\cal S}^{\dagger} \matB} = \frac{m-k}{m-n} \cdot \left( 1 + \frac{n}{k+\alpha^2}\right) \cdot \FNormS{\matB - \matB_{n}}.$$
\end{proposition}

\begin{theorem}[Frobenius Norm]\label{lower:frobenius}
For any $\alpha>0$, $n$, $m$ with $m > n$, $\mod(m,n)=0$, and $m/n > 2$, and $k$ with $n \le k \le m $, there is
a full rank $n \times m$ matrix $\matX$ such that, for any subset $\cal S$ $\subseteq [m]$
of cardinality $k \ge n$ with $\rank(\matX_{\cal S}) = n$, we have
$$ \FNormS{ \matX_{\cal S}^{\dagger} } \ge \left( \frac{m-k}{k+\alpha^2} + 1 - \frac{k}{n} \right) \cdot \FNormS{\matX^{\dagger}} $$
\end{theorem}
\begin{proof}
We construct the matrix $\matX$ as follows. Consider a block diagonal matrix $\matB \in \R^{d \times m}$:
a matrix $\matA \in \R^{d/n \times m/n}$ of the form that appears in the proof of Theorem~\ref{lower:spectral} is repeated $n$ times on $\matB$'s main diagonal. Let $\matB = \matU \matSig \matV\transp$ is the SVD of $\matB$. $\matX$ is the first $n$ rows of $\matV\transp$.

To prove the bound we use Theorem 19 and Lemma 7 from~\cite{BDM11a}.

Using $$\sigma_1(\matB) = \sigma_2(\matB) =\dots=\sigma_n(\matB)=n+\alpha^2;$$
$\sigma_{n+1}(\matB) = \sigma_{n+2}(\matB) =\dots=\sigma_{m}=\alpha^2$;
$\TNormS{\matB - \matB_n} = \alpha^2$; and $\FNormS{\matB - \matB_n} = (m-n) \alpha^2$, we obtain,
$$
\FNormS{\matB - \matB_{\cal S} \matB_{\cal S}^{\dagger} \matB} = \left(m-k\right) \cdot \left( 1 + \frac{n}{k+\alpha^2}\right)  \cdot \TNormS{\matB - \matB_{n}}.$$
Now, Lemma 7 of~\cite{BDM11a} implies that, for any matrix $\matW \in \R^{d \times m}$, rank parameter $n < \rank(\matW)$,
and sampling parameter $n \le k \le  m$, for any $\cal S$ of cardinality $k$, if $\matZ_{\cal S}$ has full rank ($\matZ$ is defined shortly),
$$
\FNormS{\matW - \matW_{\cal S} \matW_{\cal S}^{\dagger} \matW} \le \FNormS{\matW-\matW_n} + \FNormS{\left(\matW-\matW_n\right)_{\cal S}\matZ_{\cal S}^{\dagger}}.
$$
Here, $\matW_n$ is the best rank $n$ approximation to $\matW$.
$\matZ$ is defined as follows. Let $\matW = \matU \matSig \matV\transp$ be the SVD of $\matW$. $\matZ$ is the first $n$ rows of $\matV\transp$.
Applying spectral submultiplicativity to this relation we obtain,
$$
\FNormS{\matW - \matW_{\cal S} \matW_{\cal S}^{\dagger} \matW }\le  \FNormS{\matW-\matW_n}  + \FNormS{\matZ_{\cal S}^{\dagger}} \cdot \TNormS{\matW-\matW_n}.
$$
We now apply Lemma 7 from~\cite{BDM11a} on $\matB$ and $\matX$ and combine it with the bound from Theorem 19 of~\cite{BDM11a},
\eqan{
\FNormS{\matX_{\cal S}^{\dagger}} \ge \frac{\FNormS{\matW - \matW_{\cal S} \matW_{\cal S}^{\dagger} \matW}}{\TNormS{\matW-\matW_n}} - \frac{\FNormS{\matW-\matW_n}}{\TNormS{\matW-\matW_n}} 
&=&
 \left(m-k\right) \cdot \left( 1 + \frac{n}{k+\alpha^2}\right) - \left( m - n\right)  \\
 &=& 
 \left( \frac{m-k}{k+\alpha^2} + 1 - \frac{k}{n} \right) \cdot \FNormS{\matX^{\dagger}}.
}
\end{proof}

As $\alpha \rightarrow 0$ and $k=O\left(n\right)$, this bound is $m/k - O(1)$. If $k=(1 + \Omega(1))n$ the
Frobenius norm bounds in Theorems~\ref{thm1} and ~\ref{thm2} asymptotically match this lower bound.
However, if $k=(1+o(1))n$ there is a gap between the lower bound and the best upper bound.
There is also a gap when $k=\omega(n)$. We believe that the gap for $k=\omega(n)$ is the result of looseness in the lower bound,
but we were unable to prove a tighter bound than Theorem~\ref{lower:frobenius}.

\section{Low-stretch Spanning Trees and Subset Selection}\label{sec:ap}


Let $G=(V,E,w)$  be a weighted undirected connected graph.
Unless otherwise stated, in this section, we denote the number of vertices of $G$ by $n$, and the number of edges by $m$.
Let $T = (V, F, w)$ be a spanning tree of $G$, where $F$ is a subset of $E$ having exactly $n-1$ edges.
A spanning tree of a graph is a tree that spans all vertices of the given graph.
We use the same weight function $w$ because the edges in $T$ have the same weights with the corresponding edges in $G$.
Since $T$ is a tree, every pair of vertices in $T$ is connected by a unique path in $T$.
For any edge $e \in E$, let us denote by $p_{T}(e)$  the set of edges on the unique path in $T$ between the incident vertices of $e$.
The stretch of $e$ with respect to $T$ is
$\st_{\matT}(e)=\sum_{e'\in p_{T}(e)} \frac{w(e)}{w(e')}\,.$
The stretch of the graph $G$  with respect to $T$  is~\cite{AKPW95}
$$\st_{T}(G)=\sum_{e\in E}\st_{T}(e)\,.$$

The problem of finding a low-stretch spanning tree is the problem of finding a spanning tree $T$ of $G$ such that $\st_{T}(G)$  is minimized,
among all possible spanning trees of $G$. Let $\st(n)=\max_{G\in G_{n}}\min_{T}\st_{T}(G)$, where $G_{n}$  is the family of graphs with $n$  vertices. The following bounds are known: $\st(n)=\Omega(m\log n)$~\cite{AKPW95}; 
$\st(n)=O(m\log n\cdot\log\log n\cdot(\log\log\log n)^{3})$ ~\cite{ABN08}. In this section we show that finding a low-stretch spanning tree is in fact an instance of the Frobenius norm version of Problem~\ref{def:prob}.

Finding a low stretch spanning tree has quite a few uses. One important application is the solution of symmetric diagonally dominant (SDD) linear systems of equations. Boman and Hendrickson~\cite{BH02} were the first to suggest the use of low-stretch spanning trees to build preconditioners for SDD matrices. Spielman and Teng~\cite{ST04} later showed how to use low stetch spanning trees to solve SDD systems using a nearly linear amount of operations. The currently most efficient algorithm for solving SDD systems~\cite{KMP11} uses a low stretch spanning tree as well. One of the many obstacles in generalizing these algorithms for wider classes of matrices (e.g., finite-element matrices) is the lack of an equivalent concept, like the stretch, for such matrices. By studying the purely linear-algebraic nature of the low-stretch spanning tree problem (i.e. the Frobenius norm version of Problem~\ref{def:prob}),
our hope is to glean new insights on how to generalize the concept of low-stretch trees, or to substitute it with something else.

Other applications of low-stretch spanning trees include: Alon-Karp-Peleg-West game, MCT approximation and message-passing model. See~\cite{EEST05} for details.

Next, we show that finding a low-stretch spanning tree is an instance of subset selection.
We first relate graphs to matrices.
\begin{definition}[Edge-vertex incidence matrix/Laplacian matrix]
Let $G=(V,E,w)$  be a weighted undirected graph. Assume, without loss of generality, that $V=\{1,2,\dots,n\}$; $E=\{(u_1,v_1), (u_2,v_2), \dots, (u_m,v_m)\}$. The \emph{edge-vertex incidence matrix} of $G$ is $\matPi_G \in \R^{n \times m}$, where column $i$ of $\matPi_G$ is $\sqrt{w(u_i, v_i)}(\e_{u_i} - \e_{v_i})$. Here $\e_1,\e_2,\dots,\e_n$ are the identity (standard basis) vectors.
The \emph{Laplacian matrix} of $G$ is $\matL_G = \matPi_G \matPi\transp_G$.
\end{definition}
Every column in $\matPi_{G}$ represents an edge in $G$. A spanning tree $T$ is a group of edges that span $G$ and form a graph. The set of edges in $T$ correspond to a set of columns in $\matPi_G$, which we denote by ${\cal S}(T)$. Notice that if the indices are kept consistently, then, $\matPi_T = (\matPi_G)_{{\cal S}(T)}$. If ${\cal S} \subseteq [m]$ is a subset of columns, then, there is a subgraph $H$ of $G$ that contains the edges corresponding to the columns in ${\cal S}$. We denote this subgraph by $H({\cal S})$.
We are now ready to connect low-stretch spanning trees and subset selection.
\begin{theorem} \label{thm:st-to-subset}
Let $G$ be a weighted undirected connected graph. Let $\matPi_G = \matU \matSig \matV\transp$ be the SVD of $\matPi_G$
with $\matU \in \R^{n\times(n-1)}$, $\matSig \in \R^{(n-1)\times(n-1)}$ and $\matV \in \R^{m\times(n-1)}$
($G$ is connected; so, $\rank(\matPi_G) = n-1$). For notational convenience, let $\matY = \matV\transp$.
\begin{enumerate}
\item If $T$ is a spanning tree of $G$, then, $\st_{T}(G)=\FNormS{\matY^{-1}_{{\cal S}(T)}}$.
\item If ${\cal S} \subseteq [m]$ has cardinality $n-1$ and $\matY_{\cal S}$ has full rank, then, $H({\cal S})$ is a spanning tree of $G$.
\end{enumerate}
\end{theorem}
\begin{proof}
To prove the first part of Theorem~\ref{thm:st-to-subset},
we need a result of Spielman and Woo~\cite{SW09}, who recently connected the stretch of $G$ with respect to $T$ to
the matrix $\matL_G \matL^{\dagger}_T $. More precisely, Theorem 2.1 in~\cite{SW09} shows that if $T$ is a spanning tree then $\st_{T}(G) = \trace\left( \matL_{G} \matL_{T}^{\dagger} \right)$.
Here,  $G$ is a weighted undirected connected graph and $T$ is a spanning tree of $G$.

Let us denote ${\cal S} = {\cal S}(T)$. Since $T$ is a tree we have $\rank(\matPi_T)=n-1$ (it is well known that the edge incidence matrix of a connected graph has rank $|V|-1$). Since $\matPi_T = (\matPi_G)_{\cal S}= \matU \matSig \matY_{\cal S}$, $\matY_{\cal S}$ must be full rank. Now, 
\begin{eqnarray*}
 \st_{T}(G)  \buildrel{(a)}\over{=}  \trace\left( \matL_{G} \matL_{T}^{\dagger} \right)
            &\buildrel{(b)}\over{=}  &  \trace\left( \matPi_{G} \matPi_{G}\transp \left((\matPi_G)_{\cal S} (\matPi_G)_{\cal S}\transp \right)^{\dagger} \right)\\
           &  \buildrel{(c)}\over{=} &  \trace\left( \matU \matSig^2 \matU\transp (\matU \matSig \matY_{\cal S} \matY_{\cal S}\transp \matSig\matU\transp)^{\dagger} \right) \\
           &  \buildrel{(d)}\over{=} &  \trace\left( \matU \matSig^2 \matU\transp \matU \matSig^{-1} \left( \matY_{\cal S} \matY_{\cal S}\transp \right)^{-1} \matSig^{-1}\matU\transp \right) \\
           &  \buildrel{(e)}\over{=} &  \trace\left( \matSig \left( \matY_{\cal S} \matY_{\cal S}\transp \right)^{-1} \matSig^{-1}\right)\\
           &  \buildrel{(f)}\over{=} &  \trace\left( \left( \matY_{\cal S} \matY_{\cal S}\transp \right)^{-1} \right)\\
            & = & \FNormS{\matY^{-1}_{{\cal S}(T)}}\,.
\end{eqnarray*}
\math{(a)} follows by the Spielman-Woo result.
\math{(b)} follows by replacing the Laplacian matrices with the product of their edge-incidence matrices.
\math{(c)} follows by introducing the SVD of $\matPi_G$ and the equality $(\matPi_G)_{\cal S} = \matU \matSig \matY_{\cal S}$.
\math{(d)} follows since all three matrices involved ($\matU$, $\matY_{\cal S}$, and $\matSig$) are full rank.
\math{(e)} follows since $\matU$ has orthonormal columns.
\math{(f)} follows since $\trace(\matA \matB)=\trace(\matB \matA)$.

We now prove the second part of the theorem.
For $H({\cal S})$ to be a spanning tree, it has to be a connected graph with $n-1$ edges. The last condition is met since ${\cal S}$ has cardinality $n-1$, and the number of edges in $H({\cal S})$  is equal to the cardinality of ${\cal S}$. As for connectivity, notice that $\matPi_{H({\cal S})} = \left( \matPi_G \right)_{\cal S}=\matU \matSig \matY_{\cal S}$. Now, since $\matY_{\cal S}$ has full rank, we have $\rank(\matPi_{H({\cal S})}) = |V|-1$. This directly implies that $H({\cal S})$ is connected.
\end{proof}
The algorithms that we presented in Theorems~\ref{thm1} and~\ref{thm4} in Section~\ref{sec:main} can be used to find a low-stretch spanning tree (run these algorithms on the matrix $\matY$ of the above theorem), but they are not competitive both in terms of operation count and in terms of approximation bounds. Both these algorithms can guarantee $\st_{T}(G) \leq (n-1)(m-n+2)$. The operation count is $(m^2n)$ and $O(mn^3)$, respectively.
This upper bound also holds for the easily computable maximum weight spanning tree.
Koutis et al. describe in~\cite{KMP11} an algorithm which gives the available state-of-the-art upper bound
$\st_{T}(G) \leq  O(m\log n\cdot\log\log n\cdot(\log\log\log n)^{3})$,
and has operation count of $O(m \log(n) + n \log(n)\log\log(n))$. The main reason for this gap is that our algorithms are designed for general matrices, while~\cite{ABN08,KMP11} describe a graph algorithm, which better exploits the unique structure of the problem. Nevertheless, when reinterpreting our algorithms as algorithms for constructing low-stretch spanning trees yields interesting connections that sheds light on both problems, as we discuss below.

\subsection{Low-stretch spanning trees via the greedy removal algorithm}
Theorem~\ref{thm:st-to-subset} along with Theorem~\ref{thm1} suggest a greedy removal algorithm for constructing a spanning tree with low stretch: start with a full set of edges $H=E$; then, at each iteration, find the edge $e$ such that $\FNormS{\matY^{\dagger}_{{\cal S}(H - \{e\})}}$ is minimized, and set $H\longleftarrow H - \{e\}$. Finish once $H$ has $n-1$ edges. That is, we apply the algorithm of Theorem~\ref{thm1} on $\matY$ (see Theorem~\ref{thm:st-to-subset} for the definition of $\matY$). We note that this algorithm is \emph{different} from the natural greedy removal algorithm, which would remove edges to keep the stretch of the subgraph minimal in each step.
It is possible to define the stretch $\st_H(G)$ of a subgraph $H$; we refer the reader to chapter 18 of~\cite{Peleg00} for the definition. It is also possible to show that for a spanning subgraph $H$, $\FNormS{\matY^{\dagger}_{{\cal S}(H)}} \leq \st_H(G)$ (we omit the proof), but an equality does not hold. In fact, our algorithmic results imply that it is possible to find a subgraph with $O(n)$ edges such that $\FNormS{\matY^{\dagger}_{{\cal S}(H)}} = O(m)$, but there exists a graph for which every subgraph $H$ of $O(n)$ edges we have $\st_H(G) = \Omega(m \log(n))$ (Corollary 18.1.5 in~\cite{Peleg00}).

We conducted some simple experiments with our greedy removal algorithm. In the first experiment, we used greedy removal to generate a spanning tree $T_n$ of the complete graph $K_n$ on $n$ with equal weights vertices, for $n=10,11,\dots,50$. We then computed the stretch of $T_n$. We found that $\st_{T_n} (K_n) \approx 0.6 m \log^2 n$. We then repeated this experiment with random weights on the edges of $K_n$. We found that in almost all runs, $\st_{T_n} (K_n) \approx 0.3 m \log^2 n$. These values are much better than our theoretical bounds, and are closer to what it is possible to find using state-of-the-art algorithms for low-stretch trees. These experiments, although far from exhaustive, suggest that our theoretical \emph{worst-case}
upper bounds for greedy removal are rather pessimistic for the matrices relevant to finding a low-stretch spanning tree.

\subsection{Maximum weight spanning trees and maximum volume subsets}
The volume corresponding to a set ${\cal S}$ has a very natural interpretation when ${\cal S}$ is a subset of columns in $\matPi_G$, and it corresponds to a tree. Let $\matPi_G = \matU \matSig \matV\transp$ be the SVD of $\matPi_G$ ($\matU \in \R^{n\times(n-1)}$, $\matSig \in \R^{(n-1)\times(n-1)}$, and $\matV \in \R^{m\times(n-1)}$; $G$ is connected so $\rank(\matPi_G) = n-1$). For notational convenience, let $\matY = \matV\transp$. Define $\bar{\matU}$ to be the first $n-1$ rows of $\matU$; $\bar{\matU}$ is a square matrix). Define $\bar{\matPi}_G$ to be the first $n-1$ rows of $\matPi_G$. Notice that $\bar{\matPi}_G = \bar{\matU} \matSig \matY$. This implies that $\volsamp(\matY, k)=\volsamp(\bar{\matPi}_G, k)$ and also that the set ${\cal S}$ that maximizes $\det(\matY_{\cal S})^2$ also maximizes $\det((\bar{\matPi}_G)_{\cal S})^2$.

Let $T$ be a spanning tree of $G$. Determinants of the form $\det((\bar{\matPi}_G)_{{\cal S}(T)} (\bar{\matPi}_G)\transp_{{\cal S}(T)})$
have a very natural interpretation. The matrix $(\bar{\matPi}_G)_{{\cal S}(T)} (\bar{\matPi}_G)\transp_{{\cal S}(T)}$ is a Laplacian of a graph for which a column and row of some vertex have been removed. It is well known that the determinant of such matrices, when the graph is a tree, is equal to the product of the weights of the tree edges. That is,
$$
\det((\bar{\matPi}_G)_{{\cal S}(T)} (\bar{\matPi}_G)\transp_{{\cal S}(T)}) = \prod_{e \in T} w(e)\,.
$$

The subset of columns ${\cal S}$ that maximizes the volume also maximizes $\prod_{e \in H({\cal S})} w(e)$. This trivially implies that the corresponding tree is a maximum weight spanning tree. So, for edge-incidence matrices one can use an efficient maximum weight spanning tree algorithm to find the maximum volume subset of columns efficiently.
The bound we obtain is $\st_{T}(G) < (n - 1)(m-n+2)$. We are unaware of any other analysis of the stretch of a maximum weight spanning tree, but this bound can be easily proven using much simpler arguments.

\subsection{Low-stretch spanning trees via volume sampling} Recall Problem~\ref{def:prob}, and let the input matrix be the matrix $\matY$ from
Theorem~\ref{thm:st-to-subset}. Using volume sampling (Lemma~\ref{lem:volrand})
to sample a subset of columns from this $\matY$ corresponds to sampling a random spanning tree,
where a tree $T$ is sampled with relative probability $\prod_{e \in T} w(e)$. We denote this probability distribution on spanning trees of $G$ by $\Gamma(G)$.
Lemma~\ref{lem:volrand}  provides a bound on the stretch of a random spanning tree sampled from $\Gamma(G)$. Notice that it is a strict upper bound. The reason is that not every subgraph $H$ is a tree. We conjecture that this bound is pessimistic, and leave for future work the refinement of the bound.
\begin{corollary}\label{cor:randstretch}
Let $G$ be a weighted undirected connected graph, and let ${\cal T}$ be a random spanning tree, where tree $T$ is sampled with relative probability $\prod_{e \in T} w(e)$. Then,
$$
\Expect{\st_{\cal T}(G)} < (n - 1)(m-n+2)\,.
$$
\end{corollary}

One can use \textsc{VolumeSample} from~\cite{GK12} to generate such a spanning tree in $O(n^3 m)$ operations. However, the problem of generating a sample from $\Gamma(G)$ is a well studied problem, and there exists algorithms that can generate a random spanning tree faster than $O(n^3 m)$. See~\cite{Wilson96} for a short review.

\subsection{Towards better bounds for low-stretch spanning trees}
State-of-the-art algorithms for finding low stretch spanning trees attain theoretical worst-case
bounds that are better than the ones we obtain for a general matrix.
We now provide a preliminary explanation for this gap.

Consider a matrix $\matX \in \R^{n \times m}$ with orthonormal rows such that for every subset ${\cal S} \subseteq [m]$ of cardinality $n$,
$\det(\matX_{\cal S})^2 = 0$  or $\det(\matX_{\cal S})^2 = C$, for some constant $C$. The second inequality of Lemma~\ref{lem:square-to-det} is
\[
\FNormS{\matX_{\cal S}^{-1}}
=  \frac{\sum_{j=1}^{m}\sum_{i=1}^{n}\det\left(\matX_{\cal S}(i\rightarrow \x_j\right))^{2}}{\det\left(\matX_{\cal S}\right)^{2}}.
\]
(here it is an equality because $\matX$ is orthonormal; also, $\TNormS{\matX^{\dagger}} = 1$). Since all determinants are $0$ or $C$, we find that
$$
\FNormS{\matX_{\cal S}^{-1}} = \#\{{\cal T}\,:\,\rank(\matX_{\cal T})=n, {\cal T} = ({\cal S} - \{i\}) \cup \{j\} \mbox{ for } i \in {\cal S}, j \in [m]\}\,.
$$
That is, for a subset ${\cal S}$ such that the columns of $\matX$ in ${\cal S}$ form a basis for the column space of $\matX$, $\FNormS{\matX_{\cal S}^{-1}}$ is equal to the number of bases that can be obtained by replacing a single column. The last quantity can only be bounded universally by $n(m-n+1)$, and that quantity is obtained for all ${\cal S}$s if $\matX_{\cal T}$ has full rank for \emph{every} subset ${\cal T}$. However, if there exist at least one subset ${\cal T}$ for which $\matX_{\cal T}$ is singular then there is a subset ${\cal S}$ for which the $n(m-n+1)$ bound is strict. If many such sets exist, then, the bound is probably very loose.

Now, let us consider the incidence matrix of a complete graph with equal weights. If for a subset ${\cal S}$ the subgraph $H({\cal S})$ is not a tree, then, $\det(\matY_{\cal S}) = 0$ ($\matY$ is defined in Theorem~\ref{thm:st-to-subset}) . Every tree has exactly the same weight, so for all ${\cal S}$'s that correspond to trees we have the same $\det(\matY_{\cal S})^2$. We see that $\matY$ falls into the case  discussed above. We conclude that the reason that $\matY$ has a subset of column ${\cal S}$ for which $\FNormS{\matY_{\cal S}^{-1}}$ is small is the fact that for some subsets ${\cal T}$ the matrix $\matY_{\cal T}$ does not have full rank. In fact, for the complete graph, most cardinality $n-1$ subsets of edges will not result in a tree or a full rank $\matY_{\cal S}$. If we could enumerate these subsets exactly, this should give a better upper bound for this special matrix.


\section{Other Applications}\label{sec:other}

\subsection{Column-Based Low-Rank Matrix Reconstruction}\label{application:lowrank}
Suppose we want to build a low rank approximation of $\matA \in \R^{d \times m}$. For a rank parameter
$r < \rank(\matA)$, let $\matA_r \in \R^{d \times m}$ denote the best rank $r$ approximation to $\matA$. That is, $\matA_r$ minimizes $\TNorm{\matA-\matB}$, over $\matB$,
where $\matB$ ranges on all rank $r$ $d \times m$ matrices. It is well known that $\matA_r$
can be computed via the SVD of $\matA$. However,
SVD uses all the columns of $\matA$ to compute $\matA_r$. In some applications it is desirable
to use only a small set of columns to build the low rank approximation (see~\cite{BDM11a} and references there in
for such applications). Let ${\cal S} \subseteq [m]$
and $\matA_{\cal S}$ contains a subset of columns of $\matA$ indicated in $\cal S$.
Define $\matPi_{{\cal S},r}(\matA) \in \R^{d \times m}$ to be the best rank $r$ approximation of $\matA$
within the columns space of $\matA_{\cal S}$, with respect to the spectral norm
(if $\cal S$ $= [m]$, then $\matPi_{{\cal S},r}(\matA) = \matA_r$).
The so-called column-based low-rank matrix reconstruction problem is: given $\matA$, $r < \rank(\matA)$, and
an sampling parameter $k \geq r$, find a subset ${\cal S}$ of cardinality at most $k$ such that
$\TNorm{\matA - \matPi_{{\cal S},r}(\matA)}$ is minimized among all the possible choices for the subset $\cal S$.

It is natural to evaluate $\matPi_{{\cal S},r}(\matA)$ in terms of $\matA_r$. That is, provide approximation
bounds of the form $\TNorm{\matA - \matPi_{{\cal S},r}(\matA)} \le \alpha \cdot \TNorm{\matA - \matA_r}$.
Currently, the best deterministic such algorithms are available in~\cite{BDM11a}. These algorithms
achieve asymptotically optimal upper bounds, but there is still room for improvement in terms of lowering the operation count.

The algorithm of Corollary~\ref{cor1} can be used to obtain a new deterministic algorithm for the column-based low-rank matrix reconstruction problem. First, construct an SVD decomposition $\matA = \matU \matSig \matV\transp$. Let $\matX \in \R^{r \times m}$ be the first $r$ rows of $\matV\transp$. We now use the algorithm of Corollary~\ref{cor1} on $\matX$ to generate a subset ${\cal S} \subseteq [m]$ of size $k$, which is the result of the algorithm. The following bound holds,
$$ \TNorm{\matA - \matPi_{{\cal S},r}(\matA)} \le \sqrt{2 + \frac{ r(m-k) }{ k - r + 1 }  }  \cdot \TNorm{\matA - \matA_r}\,.$$
We omit the proof, which follows by combining Lemma 7 from~\cite{Bou11a} and Corollary~\ref{cor1}. The algorithm is deterministic and the operation count is $T_{SVD} + O(mr(m-k))$, where $T_{SVD}$ is the number of operations needed to compute the top $r$ right singular vectors of $\matA$.
Our approximation bound is slightly worse than the bounds in~\cite{BDM11a} but the bound on the number of operations can sometimes be better, depending on the size of the input matrix. We refer the interested reader to~\cite{BDM11a} to conduct her own comparison.

\subsection{Sparse Solutions to Least-squares Regression Problems}\label{application:sparse}
Fix inputs $\matA \in \R^{d \times m}$ and $\b \in \R^d$; consider the following least-squares problem,
$ \min_{\x \in \R^m} \TNorm{\matA \x - \b}$.
Since there is no assumption on $d$ and $m$, or that $\matA$ is full rank, the minimizer of $\TNorm{\matA \x - \b}$
might not be unique; there might be a full subspace of minimizers. Even if there is a unique minimizer, it might have
a huge norm, while there exists an almost-minimizer with small norm. It depends on the application
what exactly is needed, but often some kind of \emph{regularization} is used to address the issues just mentioned.
One popular regularization technique is \emph{truncated} SVD~\cite{Han87}: for $r < \rank(\matA)$,
let $\matA_r \in \R^{d \times m}$ of rank $r$ denote
the rank-$r$ SVD of $\matA$; then, the truncated SVD regularized
solution is given by
$
\x_{svd(r)}=\matA_r^{\dagger}\b \in \R^{m}
$.

However, sometimes a different regularization is sought: requiring the solution vector to be sparse.
That is, we are interested in constructing a vector $\x_k \in \R^m$ that has at most $k$ non-zeros, for some
$k$. Since truncated SVD is arguably the most natural
regularizer, it makes sense to compare $\x_k$ to $\x_{svd(r)}$ for some $r \leq k$. More specifically, we are interested in bounds of the form,
$$ \TNorm{\matA \x_k - \b} \le \TNorm{\matA \x_{svd(r)} - \b} + \alpha\,.$$
The idea of obtaining sparse solutions
with approximation bounds of the above type can be traced to~\cite{CH92}. Currently, the best \emph{deterministic}
method is in~\cite{Bou11b} ($k>r$) with $$\alpha = \left(1+\sqrt{\frac{r}{k}} \right) \norm{\b}_2 \FNorm{\matA-\matA_r} / \sigma_{r}(\matA).$$

The algorithm of Corollary~\ref{cor1} (Algorithm~\ref{alg2}) can be used to design a new deterministic algorithm. First, construct an SVD decomposition $\matA = \matU \matSig \matV\transp$. Let $\matX \in \R^{r \times m}$ be the first $r$ rows of $\matV\transp$. We now use the algorithm of Corollary~\ref{cor1} on $\matX$ to generate a subset ${\cal S} \subseteq [m]$ of size $k$.  We now compute $\hat{\x}_k = \matA^{\dagger}_{\cal S} \b \in \R^k$. We now form $\x_k$ as follows. For $i\in{\cal S}$ let $j_i$ be the column in $\matA_{\cal S}$ that correspond to column $i$ in $\matA$. Now, for every $i\in{\cal S}$, we set the $i$-th entry of $\x_k$ to the value of $j_i$-th entry in $\hat{\x}_k$. All other entries are set to zero.
The following bound holds,
$$
\TNorm{\matA \x_k-\b}
\le
\TNorm{ \matA \x_{svd(r)}-\b} + \left(1+\sqrt{\frac{r(m-k)}{k-r+1}} \right) \norm{\b}_2 \frac{\sigma_{r+1}(\matA)}{\sigma_{r}(\matA)}\,.
$$
We omit the proof since it follows immediately by combining Lemma 3 from~\cite{Bou11b} with
Corollary~\ref{cor1} in our paper. The algorithm is deterministic and the operation count is $O(d m \min\{d,m\} + mr(m-k))$.

Our bound essentially contains the term $\sqrt{m-k} \cdot \sigma_{r+1}(\matA)$ in place of the term $\FNorm{\matA-\matA_r}$ in the bound of~\cite{Bou11b}. It is always the case that $\FNorm{\matA-\matA_r} \le \sqrt{m-r} \cdot \sigma_{r+1}(\matA)$, but since $k \ge r$, our bound might be better in some cases (e.g. when $k \rightarrow m$).

\subsection{Feature Selection in $k$-Means Clustering}\label{application:kmeans}
The deterministic algorithm of Corollary~\ref{cor1} can also be
used for deterministic feature selection in $k$-means clustering. We refer
the reader to~\cite{BM11} for an introduction to this problem. Theorem 4 of~\cite{BM11}
gives such a polynomial-time deterministic unsupervised feature selection algorithm, which selects features from the data and then \emph{rescales} them.
Using Corollary~\ref{cor1}, one can design a deterministic unsupervised feature selection algorithm \emph{without} rescaling.
We omit the details, since the algorithm is similar to the one described for sparse least squares, and
the analysis is  a combination of Lemma 10 from~\cite{BM11} with Corollary~\ref{cor1}. The approximation bound
that is obtained is comparable to the bound in~\cite{BM11}.
\section{Open Problems and Future Directions}
Several interesting questions remain unanswered and we leave them for future investigation.
First, is the Frobenius-norm version of Problem~\ref{def:prob} NP-hard?
Second, is it possible to close the existing gaps between lower and upper bounds for Problem~\ref{def:prob}?
Third, is it possible to extend the Strong Rank Revealing QR method of~\cite{GE96} to sample arbitrary $k \ge n $ columns?
Fourth, is it possible to extend the polynomial implementations of volume sampling in~\cite{DR10,GK12} to sample arbitrary number of columns from short-fat matrices?
Finally, is it possible to derandomize the algorithm of Theorem~\ref{thm4}?

\section*{Acknowledgements}
We would like to thank the two anonymous referees and the editor for
their numerous comments and suggestions; Ioannis Koutis for bringing \cite{Packer02,Koutis06,KMP11} to our attention;
Petros Drineas, Frank De Hoog, Ilse Ipsen, Sivan Toledo, and Mark Tygert for many useful discussions and suggestions on a preliminary draft of this work;
and Anastasios Zouzias for pointing out the connection between the restricted invertibility line of research~\cite{BT87,Tropp09,SS12} and ours.

The authors acknowledge the support from XDATA program of the Defense Advanced Research Projects Agency (DARPA), administered through Air Force Research
Laboratory contract FA8750-12-C-0323.

\end{document}